\documentclass[11pt]{article}
\usepackage{tcolorbox}
\usepackage{algorithm}
\usepackage{algorithmic}

\usepackage{pifont}
\usepackage{microtype}
\usepackage{booktabs} 

\usepackage{hyperref}

\usepackage{amsfonts,amsmath,amssymb}
\usepackage{graphicx}
\usepackage{xspace}
\usepackage{url}
\usepackage{float}
\usepackage{enumerate,paralist}
\usepackage{scalerel}%
\usepackage[amsmath,thmmarks]{ntheorem}%
\usepackage{colortbl}
\usepackage{hhline}
\usepackage{multirow}
\usepackage{afterpage}
\usepackage{fullpage}
\theoremseparator{.}%

\usepackage{mleftright}%

\numberwithin{figure}{section}
\numberwithin{table}{section}
\numberwithin{equation}{section}%

\theoremstyle{plain}%
\newtheorem{theorem}{Theorem}[section]

\newtheorem{lemma}[theorem]{Lemma}

\newtheorem{corollary}[theorem]{Corollary}
\newtheorem{claim}[theorem]{Claim} 
\newtheorem{observation}[theorem]{Observation}

\newtheorem{definition}[theorem]{Definition}

\theoremstyle{plain}%
\theoremheaderfont{\bf} \theorembodyfont{\upshape}%
\newtheorem*{remark:unnumbered}[theorem]{Remark}%
\theoremheaderfont{\em}%
\theorembodyfont{\upshape}%
\theoremstyle{nonumberplain}%
\theoremseparator{}%
\theoremsymbol{\myqedsymbol}%
\newtheorem{proof}{Proof:}%

\newcommand{\myqedsymbol}{$\square$}

\newtheorem{proofof}{Proof of\!}%

\newcommand{\eps}{\varepsilon}%
\def\bar{\overline}

\def\ceil#1{\lceil {#1} \rceil}
\def\script#1{\mathcal{#1}}

\def\card#1{|#1|}
\def\set#1{\{#1\}}

\newcommand{\argmax}{\mathrm{argmax}{}\xspace}
\newcommand{\argmin}{\mathrm{argmin}{}\xspace}

\newcommand{\tldO}{\widetilde{O}}%

\def\sB{\script{B}}
\def\sC{\script{C}}

\def\sF{\script{F}}

\def\sI{\script{I}}

\def\sO{\script{O}}
\def\sQ{\script{Q}}
\def\sS{\script{S}}

\def\sU{\script{U}}
\def\sV{\script{V}}

\def\nn{\mathrm{NN}}
\def\n{\mathcal{NN}}
\def\cost{\mathrm{cost}}
\def\d{\boldsymbol{d}}

\def\ins{\mathrm{in}}

\def\opt{\textsc{OPT}}
\newcommand{\sol}{\textsc{SOL}}%
\def\Y{\boldsymbol{Y}}
\def\X{\boldsymbol{X}}

\begin{document}
\title{Individual Fairness for $k$-Clustering}
\author{
Sepideh Mahabadi\thanks{Toyota Technological Institute at Chicago (TTIC); \href{mailto:mahabadi@ttic.edu}{mahabadi@ttic.edu}} 
\and 
Ali Vakilian\thanks{Toyota Technological Institute at Chicago (TTIC); \href{mailto:vakilian@ttic.edu}{vakilian@ttic.edu}. The research was done when the author was at University of Wisconsin-Madison.}
}

\setlength{\abovedisplayskip}{3pt}
\setlength{\belowdisplayskip}{3pt}

\thispagestyle{empty}%
\setcounter{page}{0}

\pagenumbering{gobble}
\setcounter{page}{1}%
\pagenumbering{arabic}%

\date{}

\maketitle

\begin{abstract}

We give a local search based algorithm for $k$-median and $k$-means (and more generally for any $k$-clustering with $\ell_p$ norm cost function) from the perspective of individual fairness. 
More precisely, for a point $x$ in a point set $P$ of size $n$, let $r(x)$ be the minimum radius such that the ball of radius $r(x)$ centered at $x$ has at least $n/k$ points from $P$. Intuitively, if a set of $k$ random points are chosen from $P$ as centers, every point $x\in P$ expects to have a center within radius $r(x)$. An individually fair clustering provides such a guarantee for every point $x\in P$. This notion of fairness was introduced in~\cite{jung2019center} where they showed how to get an approximately feasible $k$-clustering with respect to this fairness condition.

In this work, we show how to get a bicriteria approximation for fair $k$-clustering: The $k$-median ($k$-means) cost of our solution is within a constant factor of the cost of an optimal fair $k$-clustering, and our solution approximately satisfies the fairness condition (also within a constant factor). Further, we complement our theoretical bounds with empirical evaluation.
\end{abstract}

\section{Introduction}\label{sec:intro}

Due to the increasingly use of machine learning in decision-making tasks such as awarding loans, estimating the likelihood of recidivism~\cite{galindo2000credit,chouldechova2017fair,dressel2018accuracy,kleinberg2018human}, it is crucial to design fair algorithms from the perspective of each individual input entity.
In general, the rich area of {\em algorithmic fairness} over the past few years has had two main aspects, (1) understanding different notions of fairness and formalizing them in the context of learning and optimization tasks (e.g. \cite{feldman2015certifying,kleinberg2017inherent,chouldechova2018frontiers,mehrabi2019survey,pessach2020algorithmic}) and (2) designing efficient algorithms with respect to the additional constraints caused by the fairness requirement (e.g. \cite{joseph2016fairness,hardt2016equality}). 
In this paper, we focus on the latter direction and in particular the design of fair algorithms for a basic task in unsupervised learning, namely {\em clustering}.
REcently, there has been a large body of work on design of fair algorithms for unsupervised learning and in particular clustering, e.g.,~\cite{chierichetti2017fair,bera2019fair,abraham2019fairness,har2019near, elzayn2019fair,baharlouei2019r}. 

Clustering is a fundamental task with huge number of applications such as feature engineering, recommendation systems and urban planning. Due to its importance, the clustering problem has been studied extensively from the fairness point view over the past few years~\cite{chierichetti2017fair,rosner18privacy,kleindessner2019fair,kleindessner2019guarantees,backurs2019scalable,
chen2019proportionally,schmidt2019fair,bercea2019cost,bera2019fair,huang2019coresets,
ahmadian2019clustering,davidson2019making,abraham2019fairness}. However, most previous results on this topic consider the clustering problem with respect to the notion of {\em group fairness}.
As introduced by~\cite{chierichetti2017fair}, in clustering with respect to {\em group fairness} requirement, the high-level goal is to come up with a minimum cost clustering of a given set of points with an extra constraint that requires all clusters to be balanced with respect to a set of specified protected attributes such as gender, race or religious.

In this paper, following the work of \cite{jung2019center}, we study the clustering problem from the {\em individual fairness} point of view: The goal is to design a clustering of the input point set so that all points are treated (approximately) equally. As an example, this is important when the clustering is used in certain infrastructural decisions such as where to open a new facility to serve residents in different neighborhoods. 
Formally, given a point set $P$ of size $n$, the {\em fair radius} is defined for each point $p\in P$ as the minimum radius such that the ball $B(p,r(p))$ contains at least $\ceil{n/k}$ points from $P$. Intuitively, this is the radius which the point $p$ expects to have a center within, if the centers were to be chosen uniformly at random. Therefore, it is natural to ask for a clustering solution to approximately respects this expectation and provides a clustering that has a center within distance $O(r(p))$ for every point $p\in P$, thus providing an individually fair clustering solution.

As mentioned above, this notion of fair clustering was introduced in \cite{jung2019center} where the authors showed that it is possible to get a $2$-approximate fair clustering, meaning that there exists a feasible solution of a set of $k$ centers where every point $p$ in the input has a center within distance $2r(p)$. This algorithmic result is based on the previous works of~\cite{chan2006spanners,charikar2010local} on metric embedding. Among other results, they showed that this factor of two loss in the fairness is tight: there are metric spaces and configurations of points where for $\alpha<2$ it is {\em impossible} to find $k$-centers such that for every point $p$, the distance of $p$ to its center is at most $\alpha r(p)$. Moreover, they showed empirically that the standard $k$-median ($k$-means) algorithm does not provide a good fairness guarantee and further demonstrated the {\em price of fairness} on real-world data sets.

\subsection{Our Contribution}
The above empirical result is confirmed by the following observation which shows that the solution returned by existing standard approaches for $k$-clustering that are oblivious to this fairness requirement, may be arbitrarily far from being fair.
\begin{observation}\label{lem:optimal-is-unfair}
The fairness ratio of an optimal $k$-clustering can be arbitrarily large. (Proof is in Appendix~\ref{sec:miss-intro}) 
\end{observation}
This in particular shows the need to design ``efficient'' algorithms with this notion of fairness.   

\begin{figure}[t]
\center
\includegraphics[width=0.75\textwidth]{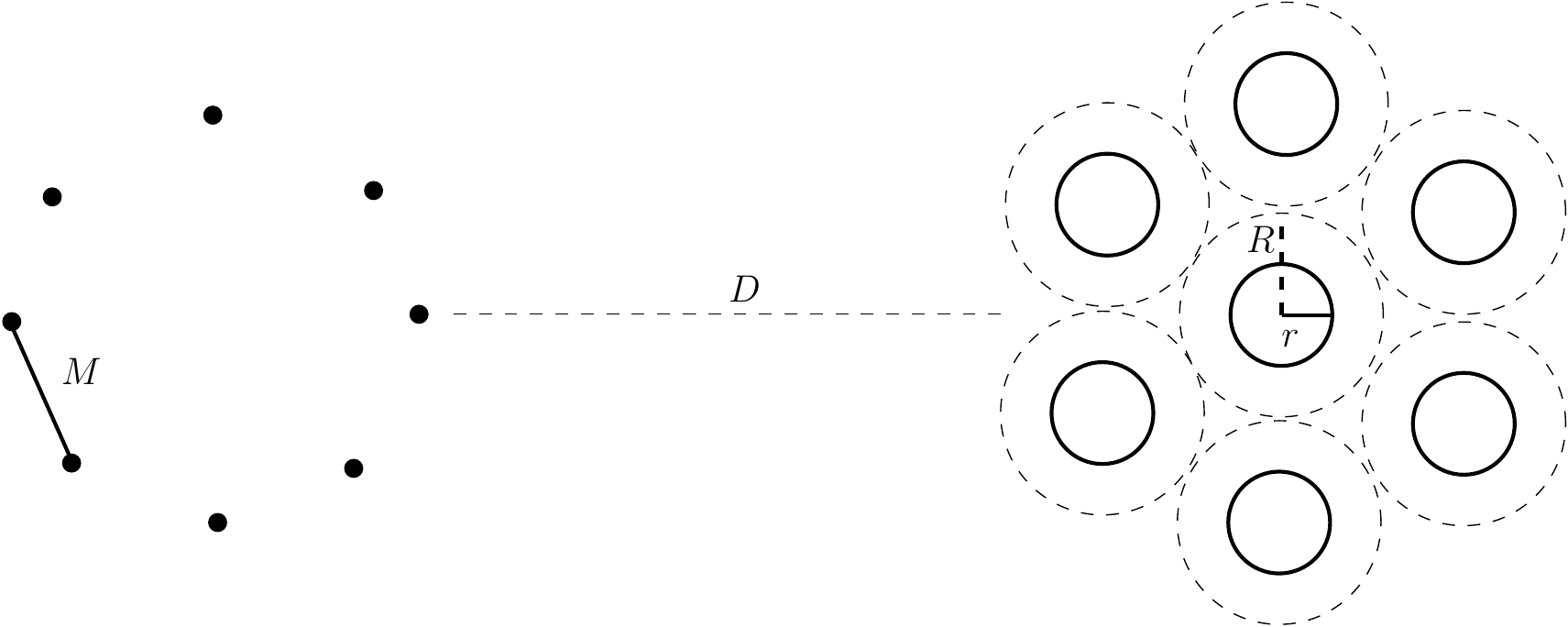}
\caption{This example shows that for any values of $\alpha$ and $k$, 
an optimal solution of $k$-median (or any other $\ell_p$ norm cost functions such as $k$-means or $k$-center) can be arbitrarily unfair.
This is described in the proof of Observation \ref{lem:optimal-is-unfair} in Appendix~\ref{sec:miss-intro}. In this example, $M$ denote the pairwise distance of points in left part and $D$ denotes the minimum distance of a point in the right part to a point in the left part. Moreover, the parameters $r, R, M$ and $D$ are picked so that $R >>r$ and $M=D >> 2n(R+r)$.}\label{fig:costly-fair}
\end{figure}

In this work, we show how to find a fair clustering that also minimizes the $k$-median or $k$-means cost (and more generally, any $\ell_p$ norm cost functions with $p\geq 1$). More specifically, we show that a variant of the local search algorithm provides the following $(O(1),O(1))$-bicriteria approximation for fair $k$-median and $k$-means.

\begin{theorem}\label{thm:main}
Given any desired fairness parameter $\alpha\geq 1$, there exists a polynomial time constant factor bicriteria approximation algorithm for $\alpha$-fair $k$-clustering: The algorithm can find a set of centers $\sC\subseteq P$ of size $k$, such that for each point $p\in P$, we have $d(p,\sC)\leq O(\alpha \cdot r(p))$, and further, $\cost(\sC)\leq O(\cost(\opt_\alpha))$. Here $\cost(\opt_\alpha)$ denotes the minimum clustering cost of $P$ with $k$ centers such that for every point $p$, there exists a center within distance $\alpha\cdot r(p)$.
\end{theorem}
We remark that while in this paper we mainly discuss $k$-median and $k$-means clusterings, more generally in Appendix~\ref{sec:general-cost}, we show that our analysis provides $(O(1), O(p))$-approximation for any cost function of the form $(\sum_{x\in P} d(x, S)^p)^{1/p}$. This in particular implies $(O(1),O(1))$-approximation for $\alpha$-fair $k$-median (by setting $p=1$ as in Corollary~\ref{cor:main-median}), $(O(1),O(1))$-approximation for $\alpha$-fair $k$-means (by setting $p=2$ as in Corollary~\ref{cor:main-means}), and $(O(1),O(\log n))$-approximation for $\alpha$-fair $k$-center (by setting $p=\log n$ in Corollary~\ref{cor:main-center}). 

We remark that while in this paper we only mention $k$-means and $k$-median clustering, more generally in Appendix~\ref{sec:general-cost}, we show that our analysis provides $(O(1), O(p))$-approximation for any cost function of the form $(\sum_{x\in P} d(x, S)^p)^{1/p}$; this in particular implies $(O(1),O(1))$-approximation for $\alpha$-fair $k$-median (when $p=1$), $(O(1),O(1))$-approximation for $\alpha$-fair $k$-means (when $p=2$) and $(O(1),O(\log n))$-approximation for $\alpha$-fair $k$-center (when $p=\log n$).

Also, we again note that our result is in contrast to the previous result which only provided one (approximately) feasible fair clustering {\em without minimizing the clustering cost}. Our algorithm is based on the local search algorithm and is easy to implement. More precisely, we start with a feasible solution which combines the output of the described algorithm of~\cite{chan2006spanners,charikar2010local,jung2019center} with the standard greedy algorithm of $k$-center (described in Section \ref{sec:initialization} in more details). Then in successive iterations, the algorithm improves the $k$-median ($k$-means) cost while respecting the fairness condition. Moreover, in order to show the theoretical guarantee, the local search algorithm should take swaps of size $4$, meaning that it considers swapping of at most $4$ centers in its current solution with the points outside of the solution.

Although local search is a widely used approach for $k$-median ($k$-means) clustering, our analysis is more involved and includes new structures that are crucial for handling the fairness constraints. 
We remark that for the sake of simplicity and readability of the paper, we have not optimized the constants in the approximation factors of the cost and the fairness guarantees.

\paragraph{Experiments.} Further, we run experiments on three datasets (Diabetes, Bank, Census) that have been previously used in the context of fair clustering (e.g., see~\cite{chierichetti2017fair,chen2019proportionally,backurs2019scalable,bera2019fair,huang2019coresets}). Our experiments show that in compare to the algorithm of \cite{jung2019center}, the $k$-median cost of our solution improves on average by a factor of 1.86, but it loses on fairness by a factor of 1.26 on average.

\subsection{Other Related Work}
Clustering is a fundamental problem in optimization and has been extensively studied in various settings.  
The $k$-center problem has a tight $2$-approximation~\cite{hsu1979easy,gonzalez1985clustering,hochbaum1985best} and after developing a series of constant factor approximation algorithms (e.g.,~\cite{charikar2002constant,arya2004local,li2016approximating}), the state-of-the-art for the $k$-median problem is $2.676$-approximation~\cite{byrka2014improved}. Also, the best known algorithm for $k$-means is a $6.357$-approximation~\cite{ahmadian2017better}. 
Refer to~\cite{aggarwal2013} for a survey on this topic.

\section{Preliminaries}\label{sec:prelim}
Throughout the paper, we use $P$ to denote the set of points that we wish to cluster, and the parameter $k$ to denote the number of centers we allow for clustering. For each $x\in P$, we use $B(x,r) = \{y\in P: d(x,y)\leq r\}$ to denote the set of points that are contained in the ball of radius $r$ around $x$. Also, in this paper we consider two main variants of clustering costs, {\em $k$-median} and {\em $k$-means}. In $k$-median, the goal is to select $k$ centers in $P$ such that the total sum of distances of points to their centers is minimized, $\min_{S\subseteq P: |S|\leq k} \sum_{p\in P} d(p,S)$, where for a set of points $S$ and a point $p$, the distance $d(p,S)$ is defined to be $\min_{s\in S} d(p,s)$. In $k$-means, the goal is to select $k$ centers in $P$ such that the sum of the square distances of points to their centers is minimized, $\min_{S\subseteq P: |S|\leq k} \sum_{p\in P}d(p,S)^2$. 

More generally, we analyze our algorithm for the fair variant of the clustering problem with any $\ell_p$ norm cost function, $\min_{S\subseteq P: |S|\leq k} (\sum_{p\in P} d(p,S)^p)^{1/p}$, where $p\geq 1$. Besides including $k$-means and $k$-median as its special cases, the general $\ell_p$ norm cost function implies an approximation guarantee for another common clustering cost function, namely $k$-center.  
In standard $k$-center, to goal is to minimize the maximum distance of points to their centers: $\min_{S\subseteq P: |S|\leq k} \max_{p\in P} d(p,S)$.   

Next, we formally define a fair radius for each point in the point set.
\begin{definition}[fair radius]
Given a set of $n$ points $P$ in a metric space $(X,d)$ and a parameter $\ell\in[n]$, for each $x\in P$ define $r_\ell(x)$ to be the radius of the minimum ball centered at $x$ that contains $({n/ \ell})$ points of $P$; $r_{\ell}(x) = \min(r: \card{B(x,r)} \geq n/\ell)$.
\end{definition}

\begin{definition}[$\alpha$-fair clustering~\cite{jung2019center}] Given a set of $n$ points $P$ in a metric space $(X, d)$, a $k$-clustering using a set of centers $S$ is $(\alpha,\ell)$-fair if for any $x\in P$, $d(x, S) \leq \alpha\cdot r_{\ell}(x)$ where $d(x, S)$ denotes the distance of $x$ to its closest neighbor in $S$. In the case $\ell = k$, we succinctly denote it as $\alpha$-fair $k$-clustering.
\end{definition} 

\begin{definition}[bicriteria approximation] 
Given a set of points $P$ in a metric space $(X, d)$, an algorithm is a ($\beta, \gamma$)-approximation for $\alpha$-fair $k$-clustering of a given cost function $\cost$\footnote{E.g., $k$-median, $k$-means, $k$-center or more generally any $\ell_p$ norm cost function.}  if the solution $\sol$ returned by the algorithm satisfies the following properties:
\begin{enumerate}
\item{\bf Cost guarantee:} $\cost(\sol) \leq \beta \cdot \cost(\opt_{\alpha})$ where $\opt_{\alpha}$ denotes an optimal solution of $\alpha$-fair clustering with respect to the given cost function $\cost$, and
\item{\bf Fairness guarantee:} $\sol$ is a $(\gamma \cdot \alpha)$-fair $k$-clustering.
\end{enumerate}
\end{definition}
Next, we define {\em critical balls} which are crucial in our analysis of the local search algorithm. The notion is used to show that our solution satisfies the fairness guarantee.
\begin{definition}[critical balls]\label{def:critical-balls}
Given a collection of $n$ points $P$ in a metric space $(X, d)$, a set of balls $B_1 = B(c^*_1, \alpha r_k(c^*_1)), \cdots, B_\ell = B(c^*_\ell, \alpha r_k(c^*_\ell))$ (where $\ell\leq k$) are called {\em critical} if they satisfy the following properties:

\begin{enumerate}[{C}-1]
	\item\label{enum:critical-1} For each $x\in P$, $d(x, \{c^*_1,\cdots, c^*_\ell\})\leq 6\alpha r_{k}(x)$,
	\item\label{enum:critical-3} For any pair of centers $c^*_i$ and $c^*_j$, $d(c_i,c_j)> 6\alpha \max\{r_k(c_i),r_k(c_j)\}$
\end{enumerate}
\end{definition}
In Lemma~\ref{lem:disjoint-balls}, we show that there exists a polynomial time algorithm for finding a set of critical balls of $P$. We say that a set of centers $S$ is {\em feasible} with respect to a set of given critical balls $\sB$ if for each $B\in \sB$, $|B\cap S|\geq 1$; each critical ball contains a center from $S$.  
\begin{claim}\label{clm:nn}
Let $o$ be a point in a critical ball $B \in \sB$.
The nearest neighbor of $o$ in a given set $S$ of feasible centers with respect to the critical balls cannot belong to a ball other than $B$. Moreover, the nearest neighbors of two points $o_1 \in B_1$ and $o_2 \in B_2$ where $B_1 \neq B_2$ cannot be the same.
\end{claim}
\begin{proof}
Consider the critical ball $B_1=B(c_1,\alpha r_k(c_1))$ that contains $o_1$. Since $S$ is a feasible center set with respect to $\sB$, $d(o_1,\nn_S(o_1))\leq 2\alpha r_k(c_1)$. However the distance of $o_1$ to any other critical ball $B_2=B(c_2,\alpha r_k(c_2))$ is at least $d(c_1,c_2) - \alpha r_k(c_1) - \alpha r_k(c_2) > 4 \alpha r_k(c_1)$ using Definition \ref{def:critical-balls}. Therefore, the nearest neighbor of $o_1$ cannot be in any ball other than $b_1$.

Now if two points $o_1\in B_1$ and $o_2 \in B_2$ are in different balls and have the same nearest neighbor $s \in S$, it means that their distance is at most $d(o_1,s)+d(o_2,s) \leq 2\alpha r_k(c_1) + 2\alpha r_k(c_2) \leq 4\alpha \max\{r_k(c_1),r_k(c_2)\}$. However by the previous argument their distance is larger than $4\alpha r_k(c_1)$ which is a contradiction.
\end{proof}

Lastly, for a point $x$, let $\nn_Y(x)$ denote the nearest neighbor of $x$ in the set of points $Y$. When $Y = P$, we drop the
subscript and refer to $\nn_P$ as $\nn$.

\section{High Level Description of Our Algorithm}

In this section, we provide a high-level description of our local search algorithm for $\alpha$-fair $k$-median. We note that the algorithm for $\alpha$-fair $k$-clustering with other cost functions (e.g., $k$-means or more generally any $\ell_p$ norm cost) is almost identical but in Section~\ref{sec:initialization}, where we compute the clustering cost in the local search algorithm, instead of working with pairwise distances we consider the corresponding function of distance (e.g., distance squared for $k$-means).

\subsection{Handling Fairness Constraints via Critical Balls}
We use a slightly modified variant of the greedy approach of~\cite{chan2006spanners,charikar2010local} to find $\ell\leq k$ disjoint {\em critical balls} (see Lemma~\ref{lem:disjoint-balls}).
Then, we ignore the fairness constraint and instead run the local search algorithm with an extra requirement: Find a set of $k$ centers that are feasible with respect to the critical balls. By Lemma~\ref{lem:clustering-with-partition-fair} in Section~\ref{sec:fairness-to-partition-constraint}, these centers result an $O(\alpha)$-fair $k$-clustering.
  
\begin{algorithm}[!h]
   \caption{Finds a set of critical balls of size at most $k$}
   \label{alg:critical-balls}
\begin{algorithmic}[1]
   \STATE {\bfseries Input:} set of points $P$ of size $n$, fairness parameter $\alpha$
   \STATE{$Z \leftarrow P, \sC^*\leftarrow \emptyset$}
   \REPEAT
	   \STATE $c \leftarrow \argmin_{x\in Z} r_{k}(x)$
	   \STATE $\sC^* \leftarrow \sC^* \cup \{c\}$
	   \STATE $Z \leftarrow \{x\in Z: d(x, c) > 6\alpha \cdot r_{k}(x)\}$	
   \UNTIL{$Z \neq \emptyset$}
   \STATE{\bfseries return} $\{B(c, \alpha r_k(c)): c\in \sC^* \}$
\end{algorithmic}
\end{algorithm}

\subsection{Initialization and Local Search Update}\label{sec:initialization}
Next, we initiate the algorithm with a feasible set of centers $S_0$ with respect to the constructed set of critical balls $\sB$ by Algorithm~\ref{alg:critical-balls}. Note that the choice of initial feasible set of centers plays an important role in bounding the number of iterations required by our local search algorithm and consequently the total runtime of our algorithms. In Theorem~\ref{thm:iterations}, we show that a modified variant of the standard greedy algorithm of $k$-center can be used to find a good initial set of centers $S_0$. See part {\bf I} in Algorithm~\ref{alg:local-search}.

Then, we go through iterations and in each iteration $j$, we check whether there exists a swap of size at most $t$ (i.e., replacing $t' \leq t$ centers in the current center set $S_j$ with a set of $t'$ centers outside of $S_j$) that results in a feasible set of centers $S'$ with respect to $\sB$ whose clustering cost improves upon the clustering cost with $S_j$ as centers by a factor more than $1/(1-\eps)$, i.e., $\cost(S') \leq (1-\eps)\cdot \cost(S_j)$. 
If there exits such a set $S'$, then we set $S_{j+1} = S'$ and proceed to the next iteration; otherwise, we stop the process and output the current set of centers, $S_j$, as our solution. We refer to such $S_j$ as a set of {\em $(t,\eps)$-stable} centers; there is no feasible swaps of size at most $t$ with respect to the critical balls that improve the clustering cost ``significantly''. See part {\bf II} in Algorithm~\ref{alg:local-search}.

\begin{algorithm}[!h]
   \caption{local search algorithm w.r.t. critical balls}
   \label{alg:local-search}
\begin{algorithmic}[1]
   \STATE {\bfseries Input:} set of points $P$ of size $n$, set of critical balls $\sB$ along with their centers $\sC^*$, upper bound $t$ on swap size 
   \STATE {\bfseries I. Constructing Initial Center Set $S'$} \label{line-a}
   \STATE $S' \leftarrow \sC^*, \ell \leftarrow |S'|$
   \FOR{$i=1$ {\bfseries to} $k-\ell$}
   	\STATE $z\leftarrow \argmax_{x\in P\setminus S'} d(x, S')$
   	\STATE $S' \leftarrow S' \cup \{z\}$
   \ENDFOR
   \STATE{\bfseries II. Local Search Update} \label{line-b}
   \REPEAT
	   \STATE $S \leftarrow S'$
	   \FOR{$i=1$ {\bfseries to} $t$}
	      \FOR{$T_1 \subseteq S$ and $T_2 \subseteq P\setminus S$ of size $i$}
	            \IF{$(S\cup T_2) \setminus T_1$ is feasible w.r.t $\sB$}
	               \STATE{$S' \leftarrow (S\cup T_2) \setminus T_1$}
	               \IF{$\cost(S') \leq (1-\eps)\cdot \cost(S)$}
					\STATE{\bfseries break to line 21}
	               \ENDIF
      		   \ENDIF
	      \ENDFOR		
	   \ENDFOR
   \UNTIL{$\cost(S') \leq (1-\eps)\cdot \cost(S)$}\label{line:update-end}
   \STATE{\bfseries return} $S$
\end{algorithmic}
\end{algorithm}

\begin{theorem}\label{thm:iterations}
The local search algorithm stops after $O({\log n \over \eps})$ iterations. Moreover, the runtime of the algorithm is $\tldO((kn)^{t+1}\cdot \eps^{-1})$.
\end{theorem}
\begin{proof}
Let $\sC^* = \{c^*_1,\cdots, c^*_\ell\}$ denote the set of the balls returned by Algorithm~\ref{alg:critical-balls}. To construct $S_0$, we run the following modified greedy algorithm of $k$-center (see part {\bf I} in Algorithm~\ref{alg:local-search}).

Initialize $S_0$ to $\sC^*$. Then, go through $k-\ell$ rounds and in each round $i$, add to $S_0$ the point $x_i \in P\setminus S_0$ to $S_0$ who is the furthest from $S_0$; $\forall x\in P\setminus S_0, d(x_i, S_0) \geq d(x, S_0)$.

When the greedy algorithm terminates, let $\mu$ denote the maximum distance of a point in $P\setminus S_0$ to $S_0$; $\mu := \max_{x\in P\setminus S_0} d(x, S_0)$. Note $S_0$ is a feasible set of centers with respect to the critical balls $\sB$ and the clustering cost with $S_0$ as centers is at most $n\cdot \mu$. Next, we show that the cost of $k$-clustering with any feasible set of centers with respect to the critical balls (in particular, the set of optimal centers $O$) is at least $\Omega(\mu)$. Consider the $k+1$ points  in $T := S_0 \cup \{x\}$. In any $k$-clustering at least two points $p_1, p_2$ in $T$ belong to a same cluster. There are two cases for $p_1$ and $p_2$:
\begin{itemize}
\item{\it At least one of $p_1, p_2$ belongs to $T\setminus \sC^*$.} Let assume that $p_2$ denote the one that added to $T$ later than $p_1$ (the last point added to $T$ is $x$). Note that by the description of the greedy process, for each $x_i$, $d(x_i, \sC^*\cup\{x_1, \cdots,x_{i-1}\}) \geq d(x, S_0) = \mu$. This implies that $d(p_1, p_2) \geq \mu$. Then, since $d(\cdot)$ satisfies the triangle inequality\footnote{Note that the squared distance (which is used in $k$-means) does not satisfy the triangle inequality but instead satisfy an approximate version of triangle inequality which is sufficient for our purpose: $d(x,y)^2 \leq 2 d(x,z)^2 + 2 d(z,y)^2$.}, the distance of at least one of $p_1$ and $p_2$ to the center of their cluster is $\Omega(\mu)$. Hence, $\opt = \Omega(\mu)$.
\item{\it Both $p_1, p_2$ belong to $\sC^*$.} Since all points in $\sC^*$ belong to different critical balls, by Claim~\ref{clm:nn}, in any feasible set of centers with respect to $\sB$,	$\sC^*$ belong to different clusters. Hence, this case cannot happen.  
\end{itemize}
We showed that $\cost(S_0) = O(n \cdot \cost(O)) = O(n\cdot \opt)$. Since, in each round of the local search algorithm the cost decrease by a factor of $(1-\eps)$, the total number of rounds before obtaining a $(t,\eps)$-stable set of centers is $O({\log n \over \eps})$. Moreover, the time to decide whether there exist a swap of size at most $t$ that improves the clustering cost by at least a factor of $(1-\eps)$ is $O(k^t \cdot n^t \cdot n\cdot k)$ where $O(k^t \cdot n^t)$ bounds the number of swaps of size at most $t$ and $kn$ denotes the required amount of time to recompute the clustering cost for each potential swap. In total, the algorithm runs in time $\tldO((kn)^{t+1} \cdot \eps^{-1})$.   
\end{proof}
\section{Handling Fairness Constraints}\label{sec:fairness-to-partition-constraint}

To satisfy the fairness requirement, as a first step, by slightly modifying the greedy algorithm of~\cite{chan2006spanners,charikar2010local}, in Lemma~\ref{lem:disjoint-balls} we show that given a set of points $P$ we can find a set of critical balls in polynomial time. Then, in Lemma~\ref{lem:clustering-with-partition-fair} we show that a feasible set of centers with respect to the critical balls is an approximate $\alpha$-fair solution.

\begin{lemma}\label{lem:disjoint-balls}
Given a set of $n$ points $P$ in a metric space $(X,d)$, there is an algorithm that runs in $O(n^2)$ and finds a set of critical balls of size at most $k$. 
\end{lemma}
\begin{proof}
Sort the points in $P$ in a non-decreasing order based on their fair radius $r_k(\cdot)$. Initially all points in $P$ are {\em uncovered}.
While there exists a point that is uncovered, we choose an uncovered point $p$ that has the smallest $r_k(p)$ and add a ball centered at $p$ of radius $r_k(p)$. Next we mark all uncovered $p'\in P$ such that $d(p,p')\leq 6\alpha r_k(p')$ as covered. Then we proceed to the next iteration. See Algorithm~\ref{alg:critical-balls} for pseudocode of this subroutine.

Note that the first property (i.e., {\it C-\ref{enum:critical-1}}) is clearly satisfied. To show that the second property (i.e., {\it C-\ref{enum:critical-3}}) holds, consider an arbitrary pair of centers $c_i$ and $c_j$. W.l.o.g., suppose that $c_i$ is added before $c_j$ which implies that $r_k(c_i) \leq r_k(c_j)$.
At the time $c_j$ is added to $\sC^*$, since it is yet uncovered, its distance to all the previously chosen centers, which includes $c_i$, is at least $6\alpha r_k(c_j)$. Lastly, since the constructed balls are disjoint and each contains at least ${n/k}$ balls, the number of critical balls is at most $k$.

The algorithm spends $O(n^2)$ to compute the fair radius values $r_k(x)$ for all $x\in P$. Then, the rest of the algorithm can be implemented in $O(n \log n + k n)$: We sort the points based on their fair radius values which takes $O(n\log n)$ and in each iteration we compute the distance of the newly picked center to all yet uncovered points in $P$ which in total takes $O(kn)$.
\end{proof}

Next, we show that in order to provide the fairness guarantee (approximately), it suffices to only satisfy the guarantee for the subset of points generated by Algorithm~\ref{alg:critical-balls}. In particular, this reduces the problem of $\alpha$-fair $k$-clustering to an instance of $k$-clustering with {\em partition} constraints: Given a collection of points $P$ and a collection of $\ell$ critical balls $\sB = \{B_1, \cdots, B_\ell\}$ where $\ell \leq k$, find a minimum cost clustering which is feasible with respect to $\sB$. 

Note that by the definition of $\alpha$-fairness, it is straightforward to verify that an $\alpha$-fair $k$-clustering of $P$ is a feasible $k$-clustering with respect to any collection of critical balls. In the following lemma we prove that any feasible $k$-clustering with respect to a collection of critical balls is $O(\alpha)$-fair.

\begin{lemma}\label{lem:clustering-with-partition-fair}
Given a set of $\ell\leq k$ of critical balls $\sB$ centered at $\sC^*$ of the input point set $P$, let $S = \set{s_1,\cdots, s_k}$ be a feasible $k$-clustering solution with respect to $\sB$.
Then, the corresponding clustering using $S$ as centers is an $O(\alpha)$-fair $k$-clustering of $P$. 
\end{lemma}
\begin{proof}
Consider a point $x\in P$ and let $c_x$ be the first center in $\sC^*$ that covers $x$. Note that since we are adding centers to $\sC^*$ in a non-decreasing order of the fair radius values,
\begin{align}\label{eq:radius-monotonicity}
r_k(c_x) \leq r_k(x)
\end{align} 
Moreover, let $s_x$ be a center in $S$ that belongs to $B_x = B(c_x, \alpha r_k(c_x))$. Then,
\begin{align*}
d(x, s_x) &\leq d(x, c_x) + d(c_x, s_x) &&\rhd \text{by triangle inequality} \\
		   &\leq 6\alpha r_k(x) + \alpha r_k(c_x) &&\rhd \text{by Property \textit{C-\ref{enum:critical-1}} and $s_x \in B_x$} \\
		   &\leq 7\alpha r_k(x) &&\rhd \text{by Eq.~\eqref{eq:radius-monotonicity}}
\end{align*}
Hence, $S$ is a $(7\alpha)$-fair $k$-clustering. 
\end{proof}
\section{Analysis of Local Search Algorithm for $\alpha$-Fair $k$-Median}\label{sec:k-median}
In this section, we analyze our proposed {\em local search} algorithm for the $\alpha$-fair $k$-median clustering. 
In Section~\ref{sec:local-search-median}, we adopt the analysis of the local search algorithm for the ``vanilla'' $k$-median~\cite{arya2004local} to the fair $k$-median problem via the mapping and covering introduced in Section~\ref{sec:mapping-covering} and show that the local search approach achieves an $(O(1), O(1))$-approximation for $\alpha$-fair $k$-clustering. Similarly, in Section~\ref{sec:general-cost}, we follow the local search analysis of~\cite{kanungo2004local,gupta2008simpler} and use our new mapping and covering constructions to prove an $(O(p), O(1))$-approximation for $\alpha$-fair $k$-clustering with respect to the general $\ell_p$ norm cost function.

\subsection{Bounded Mapping and Covering}\label{sec:mapping-covering} 
The analysis of the local search of ``vanilla'' $k$-median (and similarly $k$-means) relies on the existence of a mapping between any set of $(t,\eps)$-stable centers and the set of optimal centers. 

We use $O$ and $S$ to respectively denote an optimal set of centers for $\alpha$-fair $k$-median of $P$ and a set of $(t, \eps)$-stable centers of $P$ with respect to the critical balls $\sB$. Note that since not all sets of $k$ points in $P$ are feasible centers with respect to $\sB$, we have to deal with extra constraints when we design a mapping $\pi: O \rightarrow S$. 

Next we list the desired properties for the mapping $\pi$ and introduce a covering $\sQ$ of the edges in $\pi$ with certain properties that help us to bound the approximation guarantee of our algorithm. For simplicity, we assume that the sets $S$ and $O$ are disjoint. We remark that if $S\cap O$ is non empty, then we can simply ignore the centers in $S\cap O$ along with their clusters in the optimal solution at no cost. 
\begin{definition}[$\Delta$-bounded mapping]\label{def:bounded-mapping}
Given a pair of sets of points $S, O$ in a metric space $(X, d)$, we say a mapping $\pi: O \rightarrow S$ is $\Delta$-bounded if it satisfies the following properties:
\begin{enumerate}[{D}-1]
	\item\label{enum:mapping-1} Each center in $O$ is mapped to exactly one center in $S$.
	\item\label{enum:mapping-2} For all $s\in S$, at most $\Delta$ centers in $O$ are mapped to $s$.
\end{enumerate}
\end{definition}

\begin{definition}[$(t, \gamma)$-bounded covering]\label{def:bounded-partition}
Let $\sB = \set{B_1, \cdots, B_\ell}$ be a set of critical balls for $P$. We define $\sQ = \set{Q_1, \cdots, Q_m}$ as a covering of all edges in $\pi: O\rightarrow S$, where each $Q_i$ is a subset of the edges $(o,\pi(o))$ such that their union covers all edges of the mapping. We refer to each $Q_i$ as a partition.
For each partition $Q\in \sQ$, let $O(Q)$ denote the set of endpoints in $Q$ that belong to $O$, and let $S(Q)$ denote the set of endpoints in $Q$ that belong to $S$. We say that the covering $\sQ$ is $(t, \gamma)$-bounded if it satisfies the following properties:
\begin{enumerate}[{E}-1]
	\item\label{enum:partition-0} Each edge $(o,\pi(o))$ appears in at least one and at most $\gamma$ partitions in $\sQ$. 
	\item\label{enum:partition-1} Each partition $Q$ is a set of at most $t$ disjoint edges. In other words, (i) for each pair of $o_1\neq o_2 \in O(Q)$, $\pi(o_1) \neq \pi(o_2)$ and (ii) $|O(Q)| = |S(Q)| \leq t$. 
	\item\label{enum:partition-2} For each partition $Q \in \sQ$ and each critical ball $B_j \in \sB$, $|\big( (B_j \cap S) \setminus S(Q) \big) \cup (B_j \cap O(Q))| \geq 1$. In words, by performing the swaps corresponding to the edges in $Q$, each ball $B_j$ will have at least one center in $(S\setminus S(Q)) \cup O(Q)$. 
	\item\label{enum:partition-3} Given a partition $Q$, for all pair $o \in O(Q)$ and $o'\in O\setminus O(Q)$, $\pi(o) \neq \nn(o')$.
\end{enumerate}
\end{definition}

\subsection{Analysis of Local Search}\label{sec:local-search-median}
In Sections~\ref{sec:mapping} and~\ref{sec:covering}, we prove the main result of this section which says that given a pair of sets of centers $O$ and $S$, we can find an $O(1)$-bounded mapping $\pi$ together with a $(O(1), O(1))$-bounded covering $\sQ$ of the edges in $\pi$. 
In the following we show that if there exist such bounded mapping and covering, then the $k$-median cost of $P$ using a $(t,\eps)$-stable set of centers $S$ is within a constant factor of an optimal $\alpha$-fair $k$-median of $P$ (i.e., clustering using $O$).    

\begin{lemma}\label{lem:constant-approximate-stable}
Consider a set of points $P$ in a metric space $(X, d)$ and let $S$ be a set of $(t, \eps)$-stable centers. Suppose that there exists a pair of $\Delta$-bounded mapping $\pi: O\rightarrow S$ and $(t, \gamma)$-bounded covering $\sQ$ of the edges in $\pi$. If $\eps \leq 1/(2\gamma k)$, then $\cost(S) \leq  2\gamma\cdot (2\Delta +1) \cdot \opt$ where $\opt$ denotes the cost of an optimal $\alpha$-fair $k$-median of $P$ (i.e., the $k$-median cost using $O$).   
\end{lemma}
\begin{proof}
Consider an arbitrary partition $Q \in \sQ$. 
Here, we bound the difference between the $k$-median cost of clustering $P$ with $S$ as centers and $S_Q = (S\cup O(Q))\setminus{S(Q)}$ as centers. 

Let $S_j\subseteq P$ denote the cluster of points that are mapped to $s_j$ in the optimal $k$-median clustering with $S$ as centers and let $O_i\subseteq P$ denote the cluster of points that are mapped to $o_i$ in the optimal $k$-median clustering with $O$ as centers. Moreover, we use $\sO_Q$ and $\sS_Q$ respectively to denote $\bigcup_{o_i\in O(Q)} O_i$ and $\bigcup_{s_j\in S(Q)} S_j$.
\begin{align}
\sum_{x\in P} d(x, S_Q) - d(x, S) 
	&\leq \sum_{x\in \sO_Q} d(x, S_Q) - d(x, S) + \sum_{x\in \sS_Q \setminus \sO_Q} d(x, S_Q)  -  d(x, S)  \nonumber \\
	&\leq \sum_{x\in \sO_Q} d(x, o_x) - d(x, s_x) \qquad\rhd \text{$s_x = \nn_S(x), o_x=\nn_O(x)$} \nonumber \\ 
	&+ \sum_{x\in \sS_Q \setminus \sO_Q} d(x, s_{o_x}) - d(x, s_x) \qquad\rhd s_{o_x}= \nn_S(o_x) \label{eq:cost_of_swap}
\end{align}
Next, we bound the the value of $d(x, s_{o_x})$ for each $x\in \sS_Q\setminus \sO_Q$.
Note that since $o_x \notin O(Q)$, by Property \textit{E-\ref{enum:partition-3}} of ($\pi, \sQ$), none of the optimal centers in $O(Q)$ is mapped to the nearest neighbor of $o_x$ in $S$. In other words, $s_{o_x} \notin S(Q)$ and in particular $s_{o_x} \neq s_j$ (see Figure~\ref{fig:single_cost}).
\begin{align}
d(x, s_{o_x}) 
	&\leq d(x, o_x) + d(o_x, s_{o_x}) \nonumber \\
	&\leq d(x, o_x) + d(o_x, s_x) \quad &&\rhd \text{by the definition of $s_{o_x}$ and since $s_x \neq s_{o_x}$} \nonumber \\
	&\leq d(x, o_x) + d(o_x, x) + d(x, s_x) \quad &&\rhd \text{by the triangle inequality} \nonumber \\
	&= 2d(x, o_x) + d(x, s_x) \label{eq:closest-stable-center}
\end{align}
Now, by summing over all partitions $Q\in \sQ$,
\begin{align}
\sum_{Q\in \sQ} \sum_{x\in P} d(x, S_Q) - d(x, S) 
&\leq \sum_{Q\in \sQ} \Big( \sum_{x\in \sO_Q} \big(d(x, o_i) - d(x, s_x) \big) + \sum_{x\in \sS_Q \setminus \sO_Q} \big( d(x, s_{o_x}) - d(x, s_x)\big) \Big) \nonumber \\
&\leq \gamma\cdot \cost(O) - \cost(S) + \sum_{Q\in \sQ} \sum_{x\in \sS_Q \setminus \sO_Q} \big( d(x, s_{o_x}) - d(x, s_x)\big) \nonumber \\
&\leq \gamma\cdot \cost(O) - \cost(S) + \gamma\cdot \Delta \cdot \sum_{x\in P} d(x, s_{o_x}) - d(x, s_x) \nonumber \\
&\leq \gamma\cdot \cost(O) - \cost(S) + \gamma\cdot \Delta \cdot \sum_{x\in P} 2d(x, o_x) \quad\rhd\text{by Eq.~\eqref{eq:closest-stable-center}} \nonumber \\
&\leq \gamma\cdot(2\Delta+1)\cdot \cost(O) - \cost(S)\label{eq:partition-cost}
\end{align}
where the first inequality follows from Eq.~\eqref{eq:cost_of_swap}. 
The second inequality holds since each $o\in O$ is an endpoint of exactly one edge in the mapping $\pi$ (by Property~\textit{D-\ref{enum:mapping-1}}) and each edge of the mapping $\pi$ appears in at least one partition and at most $\gamma$ partitions of $\sQ$ (by Property~\textit{E-\ref{enum:partition-0}}). The third inequality holds since for all $x\in P$, $d(x, s_{o_x}) - d(x, s_x) \geq 0$ and each point $s\in S$ is an endpoint of at most $\Delta$ edges in the mapping $\pi$ (by Property~\textit{D-\ref{enum:mapping-2}}) and each edge of the mapping appears in at most $\gamma$ partitions of $\sQ$ (by Property~\textit{E-\ref{enum:partition-0}}).  

Next, since $S$ is a $(t,\eps)$-stable set of centers, for each $Q\in \sQ$, $\sum_{x\in P} d(x, S_Q) - d(x, S) \geq -\eps \cdot \cost(S)$, 

\begin{align*}
-\eps \cdot |\sQ|\cdot \cost(S)  
&\leq \sum_{Q\in \sQ} \sum_{x\in P} d(x, S_Q) - d(x, S) \\
&\leq \gamma\cdot(2\Delta+1)\cdot \cost(O) - \cost(S) &&\rhd\text{by Eq.~\eqref{eq:partition-cost}}
\end{align*}
which implies that $\cost(S) \leq {\gamma(2\Delta +1) \over 1 - \eps |\sQ|} \cost(O)$. Since $|\sQ| \leq k \cdot \gamma$ and $\eps \leq 1/(2k\gamma)$, the $k$-median cost of clustering with the set $S$ as centers is at most $2\gamma\cdot(2\Delta+1)\cdot \cost(O) = 2\gamma\cdot(2\Delta+1)\cdot \opt$. 
\end{proof}
\begin{figure}[t]
\center
\includegraphics[width=0.35\textwidth]{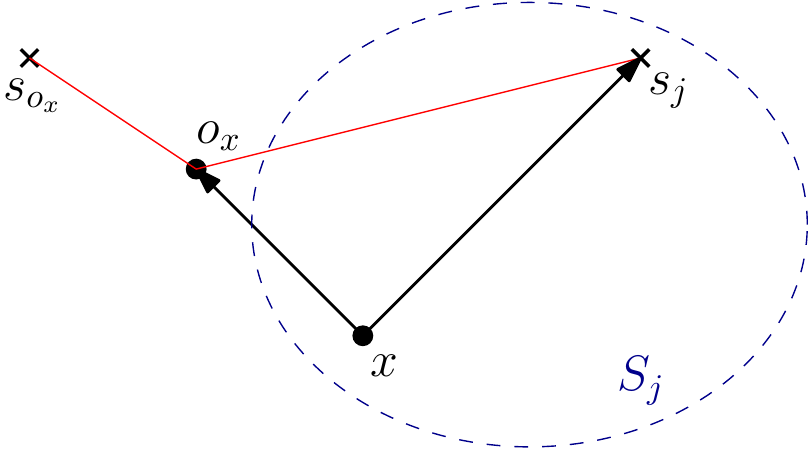}
\caption{The point $x$ belongs to $S_j$ which is a subset of $\sS_Q\setminus \sO_Q$ for a partition $Q$ in the covering $\sQ$. Then, $o_x$ does not belong to $O(Q)$ and in particular $NN_{S}(o_x)$ is different from $s_j$.}\label{fig:single_cost}
\end{figure}

\subsection{Construction of Mapping $\pi$}\label{sec:mapping}
Let us start by a few notations and claims that we will use to construct our mapping.
\paragraph{Notations.} We use $\n$ to denote the subset of points in $S$ that are the nearest neighbors of a point in $O$; $\n = \{s\in S: \exists o\in O \text{ s.t. } \nn_S(o) = s\}$. Moreover, for any value of $i\geq 0$, $\n_i$ denotes the subset of $S$ that are the nearest neighbors of exactly $i$ points in $O$; $\n_i = \{s\in S: |\{o\in O : \nn_S(o) = s\}| = i \}$. 
We also write $\n_{\geq i}$ to denote the set $\{s\in S: |\{o\in O : \nn_S(o) = s\}| \geq i \}$.

\begin{definition}[safe ball/point/edge]\label{def:safe}
A ball $B \in \sB$ is safe if 
\begin{enumerate}
\item\label{cond:ball-1} either it has at least two points from $S$, 
\item\label{cond:ball-2} or it contains a point $o\in B$ that has a unique nearest neighbor (i.e., $\nn_S(o)\in \n_1$).
\end{enumerate}
A point $s \in S$ is safe if $s \notin \n_{\geq 2}$ and
\begin{enumerate}
\item\label{cond:point-1} either $s$ is not contained in any ball of $\sB$, 
\item\label{cond:point-2} or the ball containing $s$ is safe.
\end{enumerate}
Moreover, for a pair of vertices $o\in O$ and $s\in S$ we say that the edge $(o, s)$ is safe if 
\begin{enumerate}
\item\label{cond:edge-1} either $s$ is a safe point, 
\item\label{cond:edge-2} or both $s$ and $o$ belong to the same ball in $\sB$.
\end{enumerate}
Lastly, when a ball or point or edge is not safe we denote it as unsafe.
\end{definition}

\begin{figure}[!h]
\center
\includegraphics[width=0.35\textwidth]{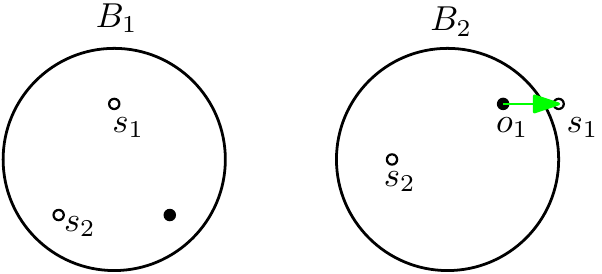}
\caption{$B_1$ and $B_2$ are safe balls.}\label{fig:safe-balls}
\end{figure}

\begin{definition}[$s_{\nn}$ and $s_{\ins}$]
For each ball $B \in \sB$, define $s_{\nn}(B)$ to be the point in $S$ that is the nearest neighbor of an arbitrary point $o\in B\cap O$. Similarly, for each ball $B\in \sB$, let $s_{\ins}(B)$ be an arbitrary point $s\in B\cap S$. 
\end{definition}
\begin{observation}
For any pair of balls $B_1, B_2 \in \sB$, $s_{\nn}(B_1) \neq s_{\nn}(B_2)$.
\end{observation}
The observation above follows from Claim~\ref{clm:nn}.
\begin{observation}\label{obr:unsafe}
The set of unsafe points that are in $\n_0$ is a subset of $\bigcup_{B\in \sB} s_{\ins}(B)$.
\end{observation}
\begin{proof}
Consider an unsafe point $s\in \n_0$. By the definition of safe points, $s$ should belong to a ball $B\in \sB$, and further (by the definition of safe balls) this ball should contain exactly one point from $S$ in it, i.e., $\{s\}=S\cap B$. But in this case $s$ is the unique choice for $s_{\ins}(B)$ and in fact $s=s_{\ins}(B)$. 
\end{proof}

\begin{observation}\label{obr:non-neighbor}
The number of points in $S$ that are not the nearest neighbor of any point in $O$ is large enough, i.e.,  
\begin{align}
|\n_0| = \sum_{s\in S} \max\{0, |\nn^{-1}(s)| - 1\}. 
\end{align}  
\end{observation}
\begin{proof}
It simply follows from the fact that $|S|=|O|$.
\end{proof}
The following lemma shows that there exist enough safe points in $\n_0$. 
\begin{lemma}
\label{lem:safe-large}
Let $\sF$ denote the set of safe points in $\n_0$. Then, 
\begin{align}\label{eq:safe-large}
|\sF| \geq \sum_{s\in S} \max\{0, |\nn^{-1}(s)| - 2\}. 
\end{align}  
\end{lemma}
\begin{proof}
Let $\sU$ denote the set of unsafe points in $\n_0$. First we prove the following statement. The number of points in $\sU$ is smaller than the number of points in $\n_{\geq 2}$:
\begin{align}\label{eq:unsafe-small}
|\sU| \leq |\n_{\geq 2}|
\end{align}
Before showing that Eq.~\eqref{eq:unsafe-small} holds, we show that Eq.~\eqref{eq:unsafe-small} suffices to prove the lemma (i.e., Eq.~\eqref{eq:safe-large}). 
By the definition, a point $s$ is in $\sF$ if $s\notin \sU \cup \n$. 

\begin{align*}
|\sF| 
&= |S| - |\n| - |\sU| \\
&\geq |S| - |\n| - |\n_{\geq 2}| &&\rhd\text{Eq.~\eqref{eq:unsafe-small}}\\
&= |S| - |\{s\in S: |\nn^{-1}(s)| \geq 1\}| - |\{s\in S: |\nn^{-1}(s)| \geq 2\}| \\
&= |O| - |\{s\in S: |\nn^{-1}(s)| \geq 1\}| - |\{s\in S: |\nn^{-1}(s)| \geq 2\}| \\
&= \sum_{s\in S} \big(|\nn^{-1}(s)| - \boldsymbol{1}_{\{|\nn^{-1}(s)| \geq 1\}} - \boldsymbol{1}_{\{|\nn^{-1}(s)| \geq 2\}}\big) \\
& = \sum_{s\in S} \max(0, |\nn^{-1}(s)|-2)
\end{align*}
Now, we show that $|\sU| \leq |\n_{\geq 2}|$.
\begin{align*}
|\sU| 
&= \sum_{B \in \sB} |\sU \cap s_{\ins}(B)|  &&\rhd\text{Observation~\ref{obr:unsafe}} \\
&= \sum_{B: s_{\nn}(B) \in \n_1} |\sU \cap s_{\ins}(B)| +  \sum_{B: s_{\nn}(B) \in \n_{\geq 2}} |\sU \cap s_{\ins}(B)| \\
&\leq \sum_{B: s_{\nn}(B) \in \n_1} |\sU \cap s_{\ins}(B)| +  | \{B: s_{\nn}(B) \in \n_{\geq 2}\} |
\end{align*}
Next, we show that for any $B$ such that $s_{\nn}(B)\in \n_1$, $s_{\ins}(B)\notin \sU$ which implies that $|\sU \cap s_{\ins}(B)| =0$. The proof is by case analysis:
\begin{itemize}
\item {$s_{\ins}(B) \in \n$:} in this case $s_{\ins}(B)$ is not in $\sU$ since by definition $\sU$ contains the unsafe points in $\n_0$.
\item {$s_{\ins}(B)\notin \n$}: in this case, since $B$ contains a point $o\in O$ with a unique nearest neighbor in $S$, by condition~\ref{cond:ball-2} of safe balls, $B$ is a safe ball and since $s_{\ins}(B)\notin \n$ the point $s_{\ins}(B)$ is safe as well.
\end{itemize}
Thus,
\begin{align*}
|\sU| 
&\leq \sum_{B: s_{\nn}(B) \in \n_1} |\sU \cap s_{\ins}(B)| +  |\{B: s_{\nn}(B) \in \n_{\geq 2}\}| \\
&= |\{B: s_{\nn}(B) \in \n_{\geq 2}\}| \\
&\leq |\n_{\geq 2}|
\end{align*}
\end{proof}

\begin{figure}[t]
\center
\includegraphics[width=0.25\textwidth]{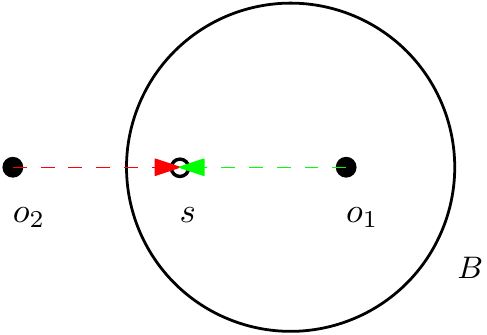}
\caption{$B$ is an unsafe ball and $s = \pi(o_1) = \pi(o_2)$. Moreover, $s$ is not the nearest neighbor of any points other than $o_1$ and $o_2$.}\label{fig:double-pointer}
\end{figure}

Next we will define a mapping $\pi: O \rightarrow S$ from which we can later derive a covering $\sQ$.
\begin{lemma}[$3$-bounded mapping]\label{lem:stable-mapping}
There exists a mapping $\pi: O \rightarrow S$ with the following properties:
\begin{enumerate}
\item{If $\nn_S(o) \in \n_1$, then $\pi(o) = \nn_S(o)$.}
\item{If $\nn_S(o) \in \n_{\geq 3}$, then the edge $(o, \pi(o))$ is safe and $\pi(o) \in \n_0$.}
\item{If $s$ is the nearest neighbor of exactly two distinct points $o_1, o_2$ in $O$, then $\pi(o_1) = \pi(o_2) \in \n_0$. Moreover, one of the following holds
\begin{enumerate}
\item each of the edges $(o_i,\pi(o_i))$ is safe. 
\item each of the edges $(o_i,\pi(o_i))$ is unsafe and goes to a ball that has a safe outgoing edge.
\item there exists an (unsafe) ball $B\in \sB$ such that $(o_1,\pi(o_1))$ belongs to $B$ and $o_2$ does not belong to any ball (see Figure~\ref{fig:double-pointer}).
\end{enumerate}
}
\item{For each $s$, $|\{o\in O: \pi(o) = s\}| \leq 3$.}
\end{enumerate}
\end{lemma}
\begin{proof}
We design $\pi$ as follows. 
\paragraph{Step 1: $\nn_S(o) \in \n_1$.} In this case we set $\pi(o) = \nn_S(o)$. This guarantees the first property in the lemma.
\paragraph{Step 2: $\nn_S(o) \in \n_{\geq 3}$.} By Lemma~\ref{lem:safe-large} there exist enough safe points in $\n_0$ so that for each point $s\in \n_{\geq 3}$, we can allocate a set of safe points $F(s)$, such that i) $|F(s)| = |\nn^{-1}(s)| - 2$, and ii) $F(s)\cap F(s')=\emptyset$ for any $s'\neq s$.
Next we define the mapping $\pi$ for the set of points in $|\nn^{-1}(s)|$ and map them to the points in $F(s)$ such that the in-degree of each safe point in $F(s)$ is at most $|\nn^{-1}(s)|/ (|\nn^{-1}(s)|-2) \leq 3$. Note that for each $o'\in \nn^{-1}(s)$, $(o', \pi(o'))$ is a safe edge. This guarantees the second property in the lemma.

\begin{claim}\label{clm:non-neighbor-count}
Let $\sF_0$ denote the set of points in $\n_0$ for which some points in $O$ are mapped to by the end of step 2. Then, $|\n_0 \setminus \sF_0| \geq |\n_{\geq 2}|$.
\end{claim}
\begin{proof}
The step 2 of the construction of the mapping consumes exactly $\sum_{s\in S} \max(0, |\nn^{-1}(s)|-2)$ safe points in $\n_0$; $|\sF_0| = \sum_{s\in S} \max(0, |\nn^{-1}(s)|-2)$. On the other hand, by Observation~\ref{obr:non-neighbor}, $|\n_0| = \sum_{s\in S} \max(0, |\nn^{-1}(s)|-1)$. Hence, $|\n_0 \setminus \sF_0| \geq |\n_{\geq 2}|$.
\end{proof}

\paragraph{Step 3.a: $\nn_S(o) \in \n_2$.} 
First we assign these points to unused safe points (if there exists any) in a way so that if $\nn(o_1) = \nn(o_2)$ then $\pi(o_1) = \pi(o_2)$. In this case both $(o_1, \pi(o_1))$ and $(o_2, \pi(o_2))$ are safe and the degree of $\pi(o_1) = \pi(o_2)$ is exactly $2$ which corresponds to property 3.(a) in the lemma.

Note that by Claim~\ref{clm:non-neighbor-count}, at the end of step 2, the number of free points in $\n_0$ is at least $|\n_{\geq 2}|$ and we have enough free points in $\n_0$ for mapping the subset of points whose nearest neighbors belong to $\n_2$. 
\begin{claim}\label{clm:safe-outgoing}
All outgoing edges of the mapping $\pi$ constructed so far (i.e., in step 1, 2 and 3.a) that leave a ball in $\sB$ are safe.
\end{claim}
\begin{proof}
All edges of $\pi$ constructed in steps 2 and 3.a are safe as stated earlier. So consider an outgoing edge $(o,s)$ in $\pi$ where $s=\nn_S(0)$ belongs to $\n_1$. By Claim \ref{clm:nn}, $s$ cannot be in any ball and thus it is a safe point and hence $(o,s)$ is safe.
\end{proof}

Next, let $\bar{\n}_2$ denote the set of points $s\in \n_2$ such that the mapping $\pi$ is not yet defined for $o_1, o_2$ where $\nn_S(o_1) = \nn_S(o_2) = s$. Moreover, let $\bar{\n}_0$ denote the subset of $\n_0$ that are not yet used in the mapping $\pi$ so far. Since for each pair $o_1, o_2$ where $\nn(o_1) = \nn(o_2) \in \n_2$, $\pi(o_1) = \pi(o_2)$, it is straightforward to verify that the invariant $|\bar{\n}_0| \geq |\bar{\n}_2|$ still holds.

\paragraph{Step 3.b: $\nn_S(o) \in \bar{\n}_2$.} 

\begin{enumerate}
\item\label{case:inner-edge} For a pair of points $o_1$ and $o_2$ whose (identical) nearest neighbor belongs to $\bar{\n}_2$, if there exists an unsafe ball in $\sB$ that contains a free point $s \in \bar{\n}_0$ (i.e., is not assigned to any point in $O$ by the mapping $\pi$ so far) and $\sB$ contains at least one of $o_1$ and $o_2$ then set $\pi(o_1) = \pi(o_2) = s$. (Note that using Claim \ref{clm:nn}, $o_1$ and $o_2$ cannot belong to different balls as they share the same nearest neighbor.)
\item\label{case:outer-edge} Once there is no such pair of points $o_1$ and $o_2$ as in the previous case anymore, we consider an arbitrary one-to-one assignment $\phi$ from the free points in $\bar{\n}_2$ to $\bar{\n}_0$. Finally, for a pair of points $o_1$ and $o_2$ where $\nn_S(o_1) = \nn_S(o_2) = s \in \bar{\n}_2$ we set $\pi(o_1) = \pi(o_2) = \phi(s)$.  
\end{enumerate}

Next we show that if $(o_1, \pi(o_1) =s_1)$ is an unsafe edge (which is true for all constructed edges in step 3.b), then either property 3.(b) or 3.(c) holds. Suppose for contradiction that $(o_1, s_1)$ goes to a ball $B$ that does not have any safe outgoing edge and $o_2 \notin B$. Note that by the feasibility of centers set $O$, the ball $B$ contains a point in $o \in O$ where $o\notin \{o_1,o_2\}$. 
Since $(o_1, s_1)$ is an unsafe edge, $B\cap S = \{s_1\}$ which implies that $\pi(o)$ does not belong to $B$. 
Moreover, by our construction (step 3.b-\ref{case:inner-edge} above) $\nn_S(o)\notin \bar{\n}_2$; otherwise, $o$ had to be mapped to $s_1$. Hence, by Claim~\ref{clm:safe-outgoing}, the outgoing edge $(o, \pi(o))$ is safe which is a contradiction.
\end{proof}

\subsection{Construction of the Covering $\sQ$}\label{sec:covering}

\begin{lemma}[$(4,6)$-bounded Covering]\label{lem:stable-partitioning}
Given a mapping $\pi$ that satisfies conditions of Lemma~\ref{lem:stable-mapping}, we can find a $(4,6)$-bounded covering $\sQ$.
\end{lemma}
\begin{proof}
First we show that if an edge $(o, \pi(o))$ is safe then there exists a feasible $2$-swap that contains the edge. 
Consider the following cases of a safe edge:
\begin{enumerate}[I.]
\item In the following cases, $(o, \pi(o))$ is a feasible (singleton) partition. 
	\begin{enumerate}
		\item $o$ and $\pi(o)$ both belong to the same ball.
		\item $\pi(o)$ does not belong to any ball.
		\item $\pi(o)$ belongs to a ball that contains at least two points from $S$.
	\end{enumerate}
\begin{figure}[t]
\center
\includegraphics[width=0.5\textwidth]{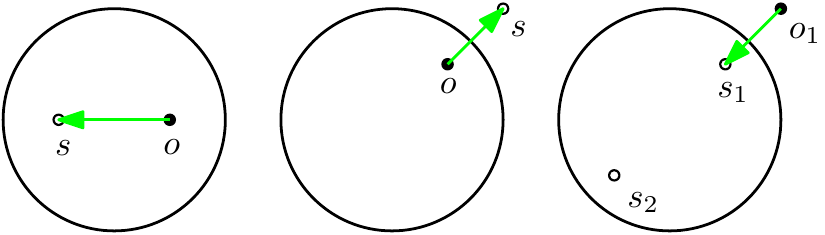}
\caption{Examples of singleton partitions.}\label{fig:partition-I}
\end{figure}
\item\label{item-second} Let $B$ denote the ball containing $\pi(o)$. Then $B\cap S =\{\pi(o)\}$ and there exists a point $o' \in B$ such that $\pi(o') \in \n_1$ and $\pi(o')\notin B$. Thus we form a partition of size two consisting of $(o, \pi(o))$ and $(o', \pi(o'))$ which corresponds to a feasible $2$-swaps. This is a valid partition because by Claim~\ref{clm:nn}, $\pi(o')$ cannot be contained in any ball.
\begin{figure}[b]
\center
\includegraphics[width=0.25\textwidth]{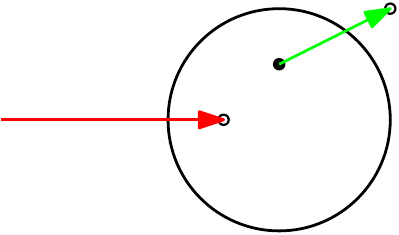}
\caption{An examples of a partition of size $2$.}\label{fig:partition-II}
\end{figure}  	
\end{enumerate}  
Note that in case~\ref{item-second}, since $B\cap S =\{\pi(o)\}$ and by the last condition in Lemma~\ref{lem:stable-mapping} the in-degree of $\pi(o)$ is at most $3$, the edge $(o', \pi(o'))$ appears in at most three $2$-swaps. 

By the properties of Lemma~\ref{lem:stable-mapping}, $(o,\pi(o))$ is a safe edge in the following cases and by what we just showed there is a decomposition of edges into swaps of size at most two such that each edge appears in at most three of them.
\begin{itemize}
\item $\nn(o) \in \n_{\geq 3}$.
\item $\nn(o) \in \n_{2}$ and $(o, \pi(o))$ is a safe edge (Case 3-(a) in Lemma~\ref{lem:stable-mapping}).
\item $\nn(o) \in \n_1$ and either the edge is fully contained in a ball or $\pi(o)$ is not inside any ball. 
\end{itemize} 
Similarly, if $\nn(o) \in \n_{2}$ and $(o, \pi(o))$ goes to a ball that has a safe outgoing edge (Case 3-(b) in Lemma~\ref{lem:stable-mapping}), then all such edges are covered by swaps of size three such that each edge appears in at most $6$ of them.

Hence, it only remains to handle the unsafe edges $(o, \pi(o))$ where either 
\begin{enumerate}[(1)]
\item $\nn(o)\in \n_2$ and there exists a ball $B\in \sB$ such that $(o,\pi(o))$ belongs to $B$ and $o_2$ does not belong to any ball where $\nn_S(o) =\nn_S(o_2)$ (Case 3-(c) in Lemma~\ref{lem:stable-mapping}), or
\item $\nn_S(o)\in \n_1$, $(o, \nn_S(o))$ is an unsafe edge, $o$ is not inside any ball and $\pi(o)$ is contained in a ball. 
\end{enumerate}
Note that the second case can be handled easily as the ball containing $\pi(o)$ has an outgoing edge and by this stage all outgoing edges are covered in a swap of size at most $3$. We can add the edge to the swap corresponding to the outgoing edge of $\pi(o)$ which will results in a swap of size at most $4$. Moreover, since $(o, \pi(o))$ is the only incoming edge of the ball containing $\pi(o)$ still each edge appear in at most $6$ swaps.  

To cover the edges of type $(1)$ above, we set $\pi(o_2) = \nn(o_2)$. Then we add $(o, \pi(o))$ and $(o_2, \pi(o_2) = \nn(o_2) = \nn(o))$ as a partition in $Q$. Note that this satisfies  all properties of $(t, \gamma)$-bounded mapping and in particular Property~\textit{E-\ref{enum:partition-3}}. Moreover, these edges appear in exactly one partition of $\sQ$. 
\end{proof}

\begin{corollary}[restatement of theorem \ref{thm:main} for $k$-median]\label{cor:main-median}
The local search algorithm with swaps of size at most $4$ returns a $(84, 7)$-bicriteria approximate solution of $\alpha$-fair $k$-median of a point set $P$ of size $n$ in time $\tldO(k^{5} n^4)$.
\end{corollary}
\begin{proof}
By Lemma~\ref{lem:clustering-with-partition-fair}, the result of our local search algorithm returns a $(7\alpha)$-fair $k$-clustering. By Lemma~\ref{lem:constant-approximate-stable} and the existence of a pair of $3$-bounded mapping and $(4,6)$-bounded covering (see Lemma~\ref{lem:stable-mapping} and~\ref{lem:stable-partitioning}), the cost the returned $k$-clustering is not more that $84\cdot \opt$ where $\opt$ is the cost of an optimal $\alpha$-fair $k$-clustering. 

Finally, as we set $\eps = O(1/k)$ and by Theorem~\ref{thm:iterations} the runtime of the algorithm is $\tldO(k^5 n^4)$.
\end{proof}

\section{Experiments}\label{sec:experiments}
In this section, we provide an empirical evaluation of our local search based algorithm for $\alpha$-fair $k$-clustering with $k$-median and $k$-means cost functions. We evaluate the empirical performance of the following approaches:
\begin{itemize}
\item{\bf \textsc{FairKCenter}~\cite{jung2019center}.} 
First we consider the algorithm of~\cite{jung2019center}, where they proposed to perform a binary search to find a value of $1\leq \eta \leq 2$ for which the number of critical balls turns out to be exactly $k$. Then, the centers of these balls are the clustering centers.\footnote{\cite{jung2019center} remarked that while it is not always guaranteed that there exists such $\alpha$ that results in {\em exactly} $k$ balls, this approach finds $k$ centers on most natural datasets.}  
More precisely, for a given value of $\eta$, they run the algorithm of~\cite{chan2006spanners,charikar2010local}, which is similar to Algorithm~\ref{alg:critical-balls} but in line 6 a point $x$ is marked as ``covered'' if $d(x,c) \leq \eta \cdot r_k(x)$ where $c$ is the newly picked center.
\item{\bf Local Search with 1-Swap.}
 Second, we consider our local search algorithm as described in Algorithm~\ref{alg:local-search}. However to make it faster, we set $t$, the maximum size of swaps, equal to $1$ instead of $4$.
\item{\bf Greedy.} Finally as it seems a reasonable heurictic, we also consider the solution after the initialization step for our local search algorithm that is described in Algorithm~\ref{alg:local-search}, part {\bf I}. This algorithm first finds critical balls and then includes their centers to $S$. Next, it goes through iterations until $S$ becomes of size $k$. In each iteration it adds the point which furthest from the current set $S$. This algorithm is the initialization method we used for our local search algorithm and is described in Algorithm~\ref{alg:local-search}, part {\bf I}.
\end{itemize}
\paragraph{Dataset.}
We consider three datasets from UCI Machine Learning Repository~\cite{Dua:2017}\footnote{\href{https://archive.ics.uci.edu/ml/datasets}{https://archive.ics.uci.edu/ml/datasets}} which are standard benchmarks for clustering algorithms and in particular they were used in the context of fair $k$-median clustering in \cite{chierichetti2017fair,chen2019proportionally,backurs2019scalable,bera2019fair,huang2019coresets}. Formally, we consider the following datasets where in each of them we consider only numerical attributes:
\begin{itemize}
\item{\it  Diabetes.} This dataset provides the information and outcome regarding patients related to diabetes from 1999 to 2008 at 130 hospitals across US\footnote{\href{https://archive.ics.uci.edu/ml/datasets/diabetes+130-us+hospitals+for+years+1999-2008}{https://archive.ics.uci.edu/ml/datasets/diabetes+130-us+hospitals+for+years+1999-2008}}. Points in this datasets are in $\mathbb{R}^2$ and correspond to ``age'' and ``time-in-hospital'' attributes.   
\item{\it Bank.} This datasets corresponds to information from a Portuguese Bank\footnote{\href{https://archive.ics.uci.edu/ml/datasets/Bank+Marketing}{https://archive.ics.uci.edu/ml/datasets/Bank+Marketing}}. Here, points live in $\mathbb{R}^3$ and corresponds to ``age'', ``balance'' and ``duration-of-account''.
\item{\it Census.} The dataset is from 1994 US Census\footnote{\href{https://archive.ics.uci.edu/ml/datasets/adult}{https://archive.ics.uci.edu/ml/datasets/adult}} and here the selected attributes are ``age" , ``fnlwgt", ``education-num", ``capital-gain" and ``hours-per-week"; points are in $\mathbb{R}^5$.
\end{itemize} 

\begin{table}[!h]
\centering
\renewcommand{\arraystretch}{.75}
\begin{tabular}{l|l|l|l}
\toprule 
 {\bf Dataset} & {\bf Dimension} & {\bf \# of Points} & {\bf Aspect Ratio}\\
\midrule
Diabetes & $2$ & $101,765$ & $90.2$ \\ 
 Bank & $3$ & $4,520$ & $13511.9$\\
 Census & $5$ & $32,560$ & $	58685$\\
\bottomrule
\end{tabular}
\caption{Some statistics about the datasets used in our experiments. Aspect ratio denotes the ratio between maximum distance and minimum distance.}
\label{table:dataset-stat}
\end{table}
Finally, in all our experiments
we randomly sample a subset of size 1000 points from the data set and run our experiments on this sub-sample.

\paragraph{Experiment Setup.} In our experiments, we follow the description of Algorithm~\ref{alg:critical-balls} and~\ref{alg:local-search}. The only discrepancy is that instead of considering a point to be covered if it has a center within distance of $6$ times its fair radius, in our implementation, we consider a point covered if it has a center within distance of $3$ times its fair radius (see line 6 in Algorithm~\ref{alg:critical-balls}). 

In all experiments, the input parameter $\alpha$ to our local search algorithms (i.e., the desired fairness guarantee) is the fairness approximation $1 \leq \eta \leq 2$ returned by the \textsc{FairKCenter} algorithm of~\cite{jung2019center}.  

Finally, we consider values of $k$ to be in range $5$ to $30$ with steps of size $5$ and draw our plots as a function of $k$.
\begin{figure*}[!h]
\minipage{0.33\textwidth}
		\includegraphics[width=\textwidth]{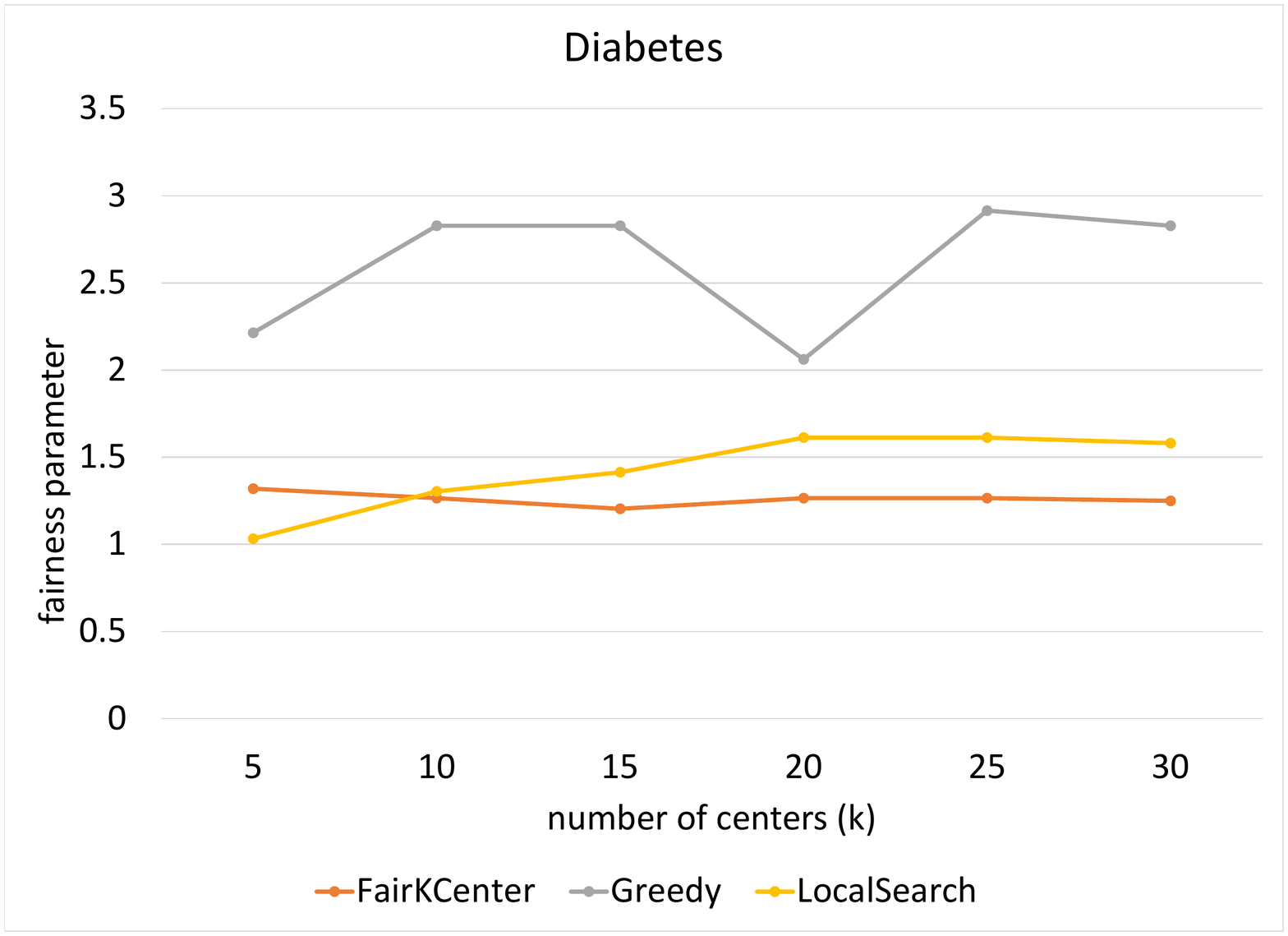}
\endminipage\hfill
\minipage{0.33\textwidth}
		\includegraphics[width=\textwidth]{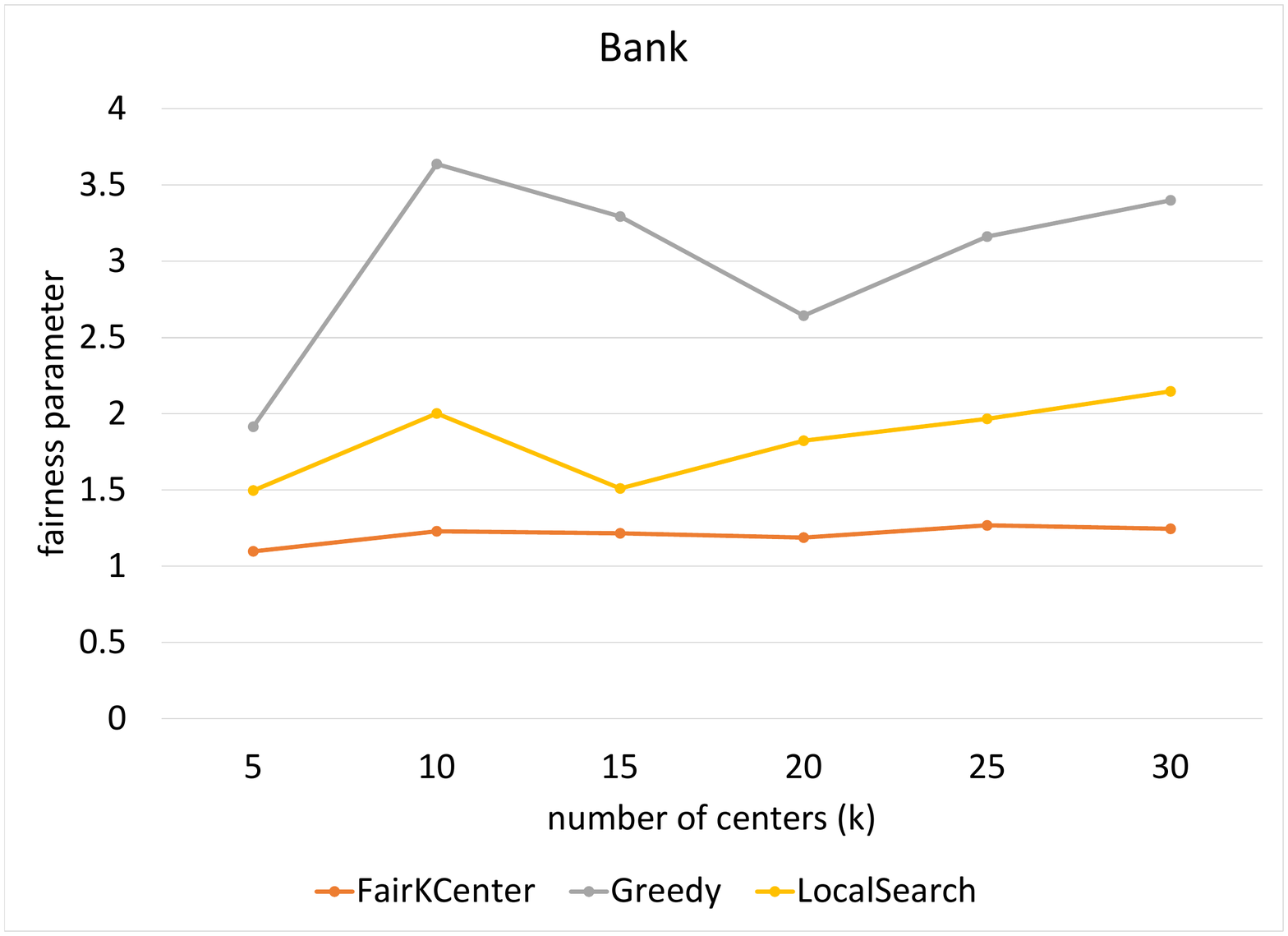}
\endminipage\hfill
\minipage{0.33\textwidth}%
		\includegraphics[width=\textwidth]{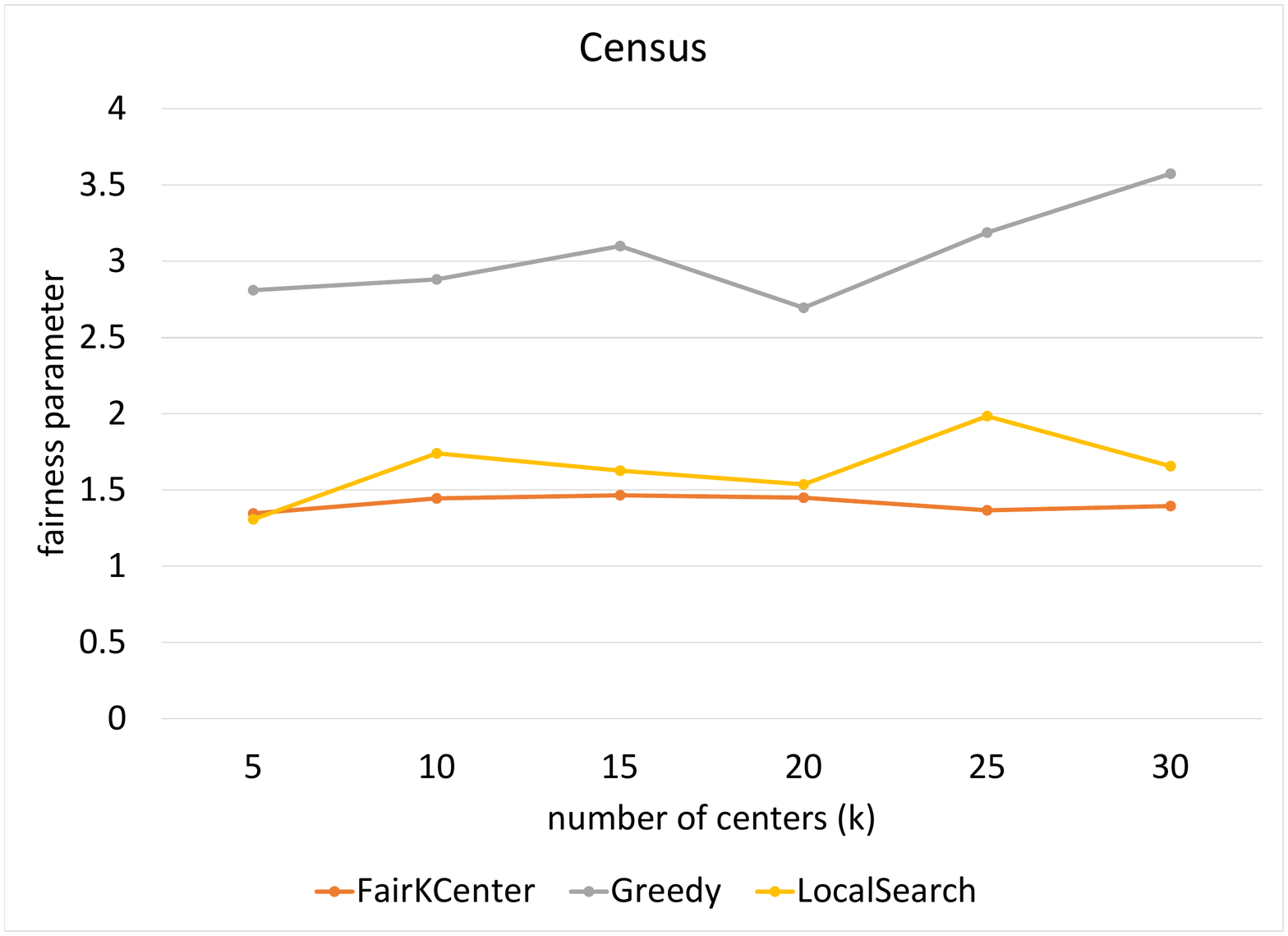}
\endminipage
\caption{Comparison of the fairness guarantees of the described algorithms for fair $k$-median on data sets Diabetes, Bank and Census.}
\label{fig:fair-median}
\end{figure*}

\begin{figure*}[!h]
\minipage{0.33\textwidth}
		\includegraphics[width=\textwidth]{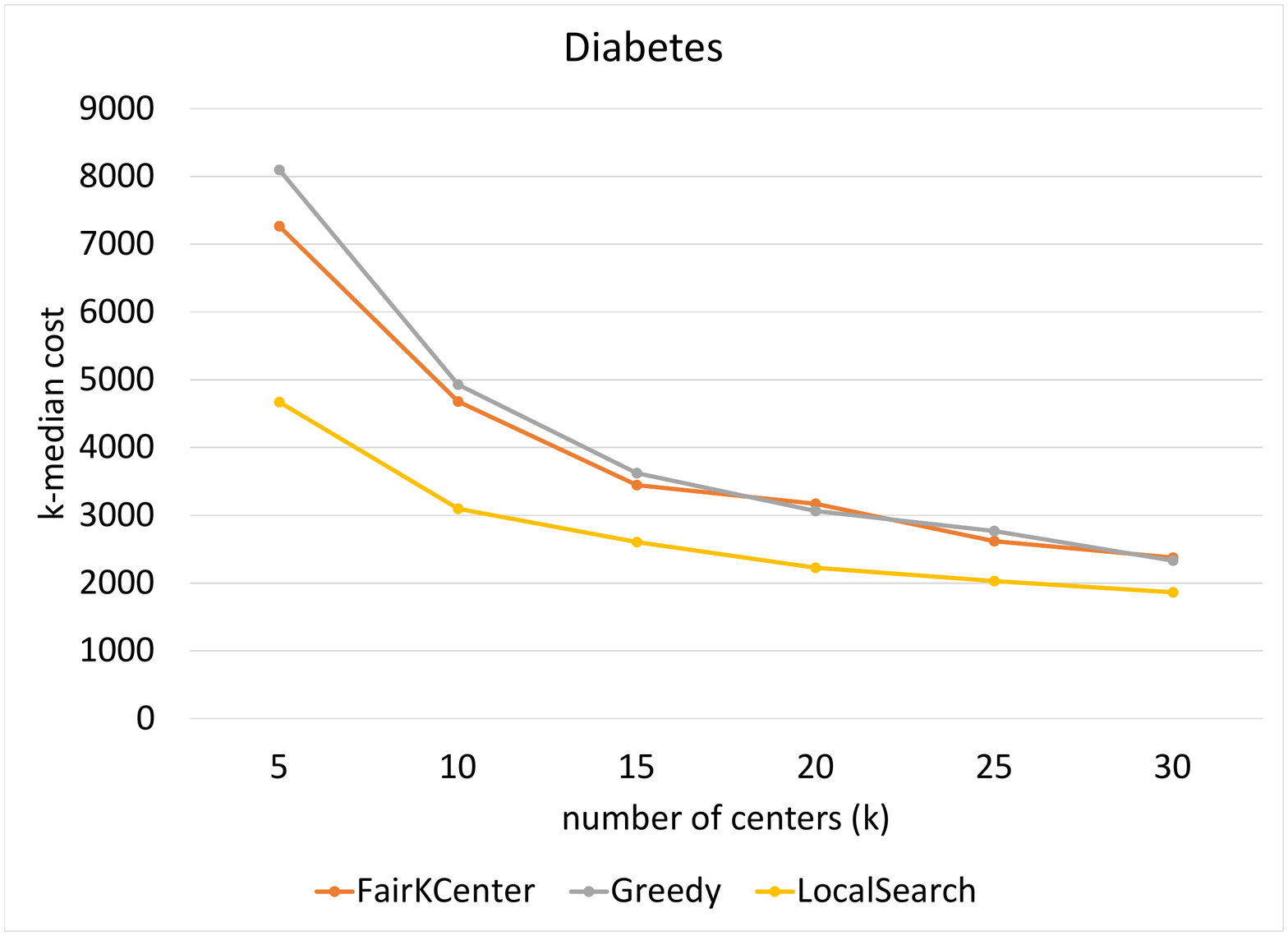}
\endminipage\hfill
\minipage{0.33\textwidth}
		\includegraphics[width=\textwidth]{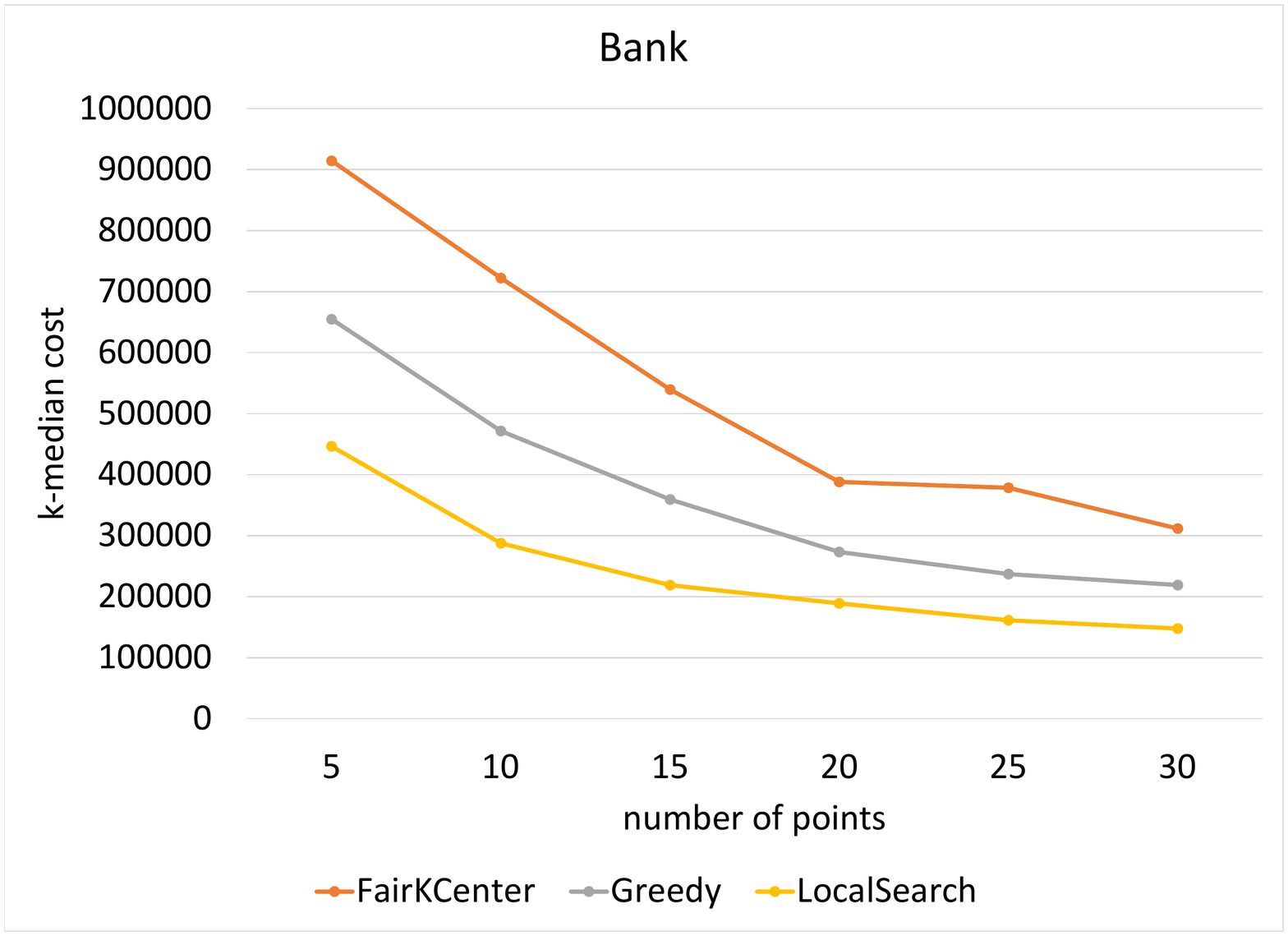}
\endminipage\hfill
\minipage{0.33\textwidth}%
		\includegraphics[width=\textwidth]{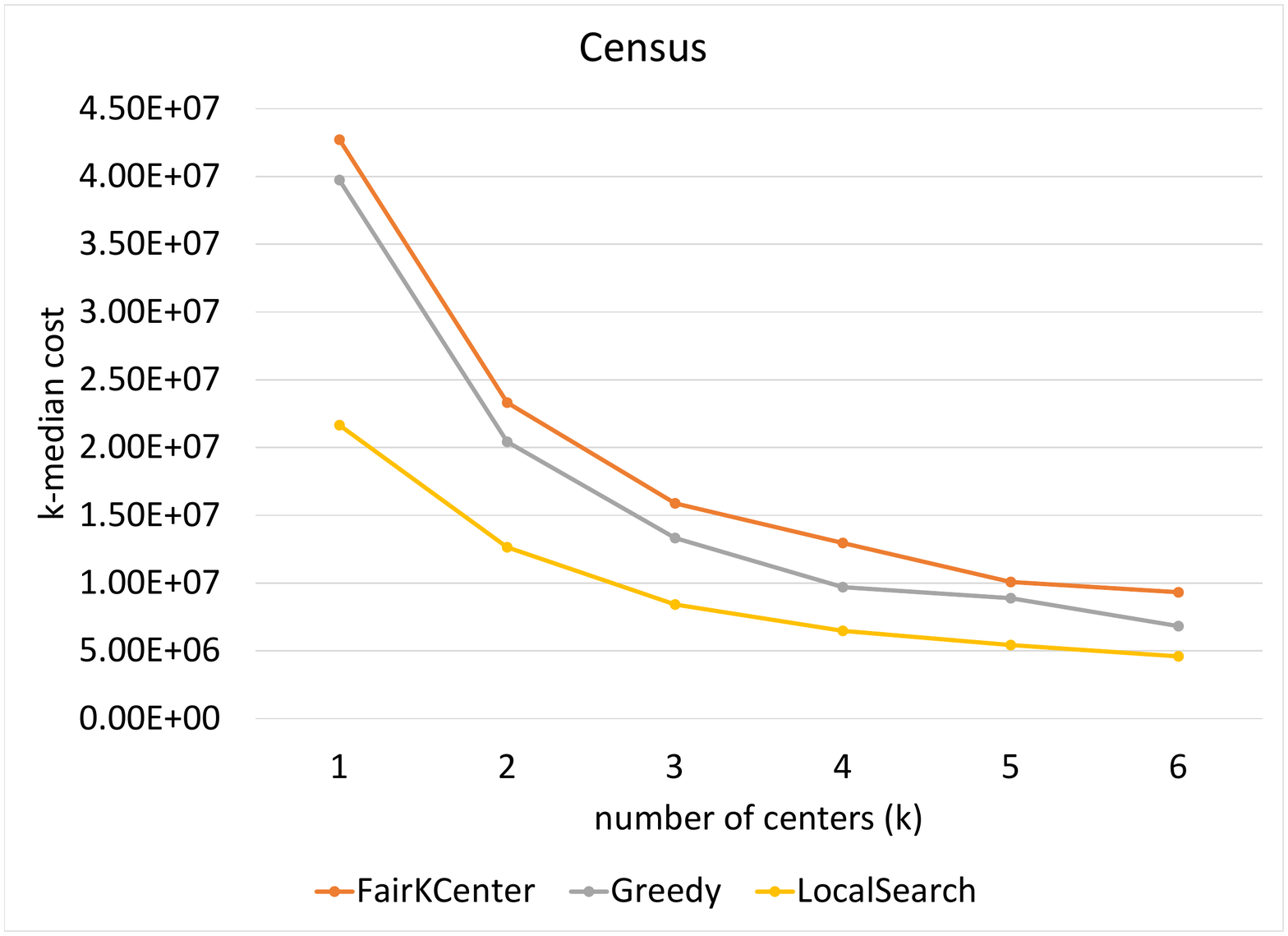}
\endminipage
\caption{Comparison of the $k$-median cost of the described algorithms for fair $k$-median on data sets Diabetes, Bank and Census.}
\label{fig:cost-median}
\end{figure*}

\begin{figure*}[!h]
\minipage{0.33\textwidth}
		\includegraphics[width=\textwidth]{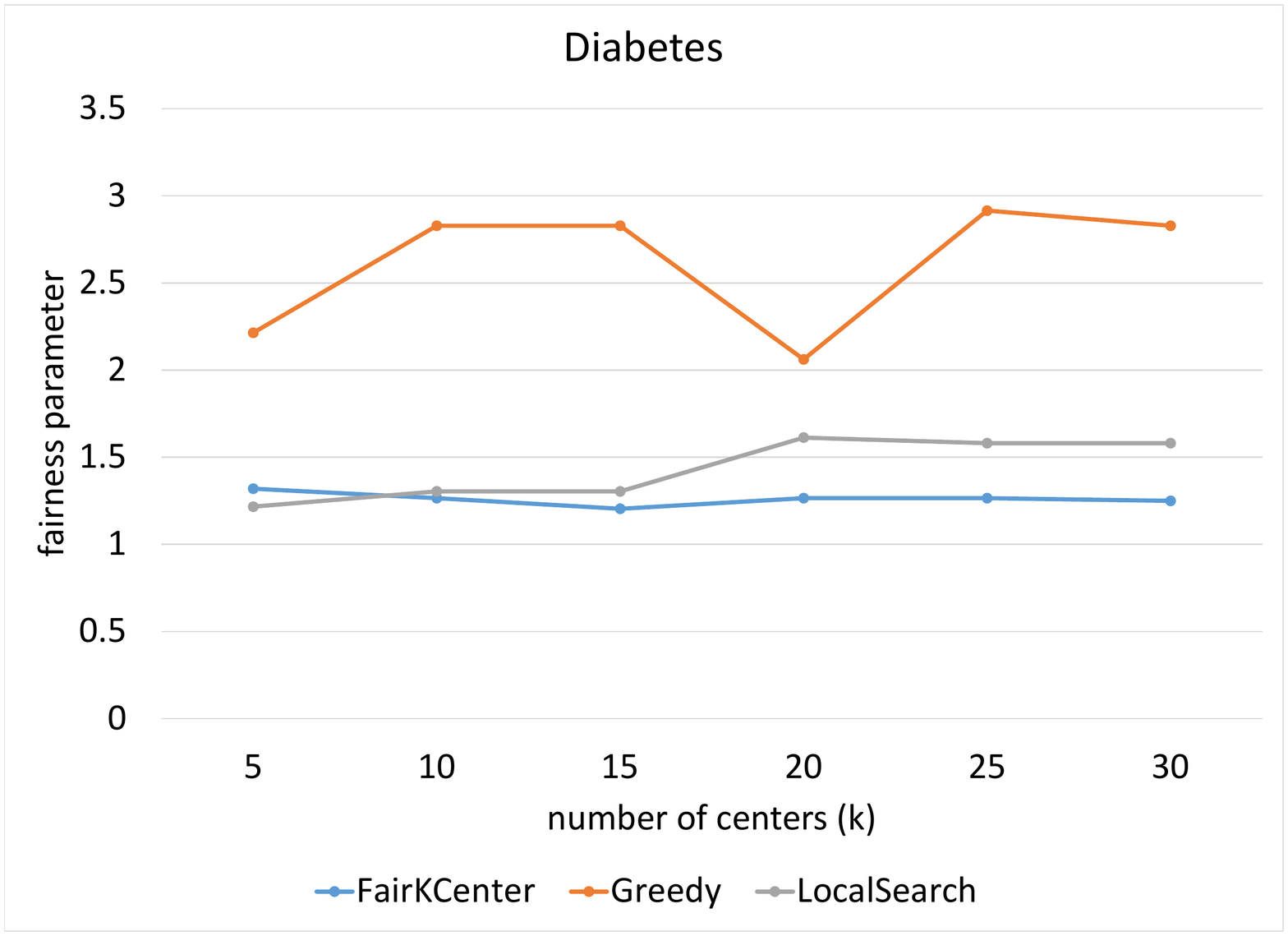}
\endminipage\hfill
\minipage{0.33\textwidth}
		\includegraphics[width=\textwidth]{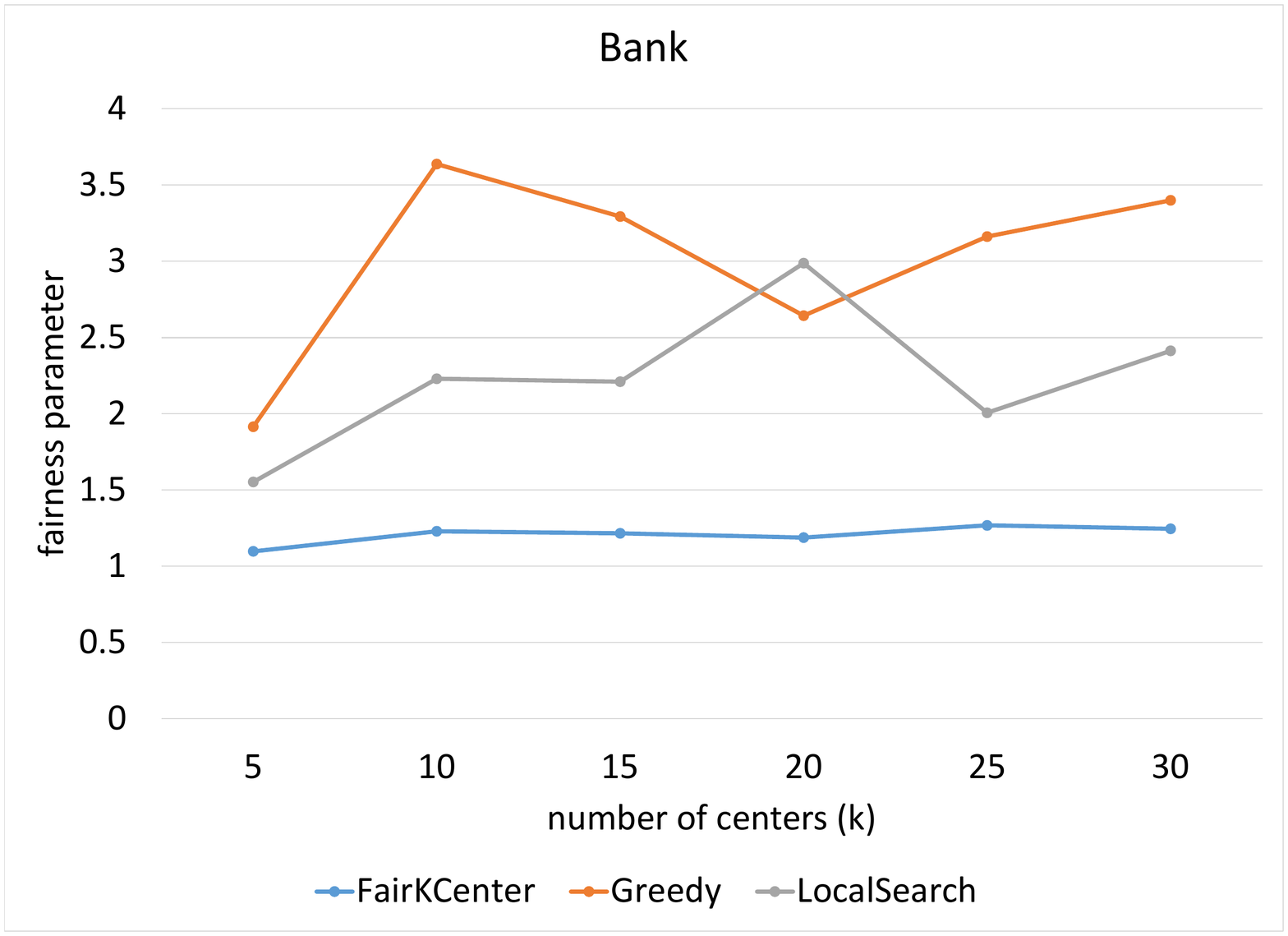}
\endminipage\hfill
\minipage{0.33\textwidth}%
		\includegraphics[width=\textwidth]{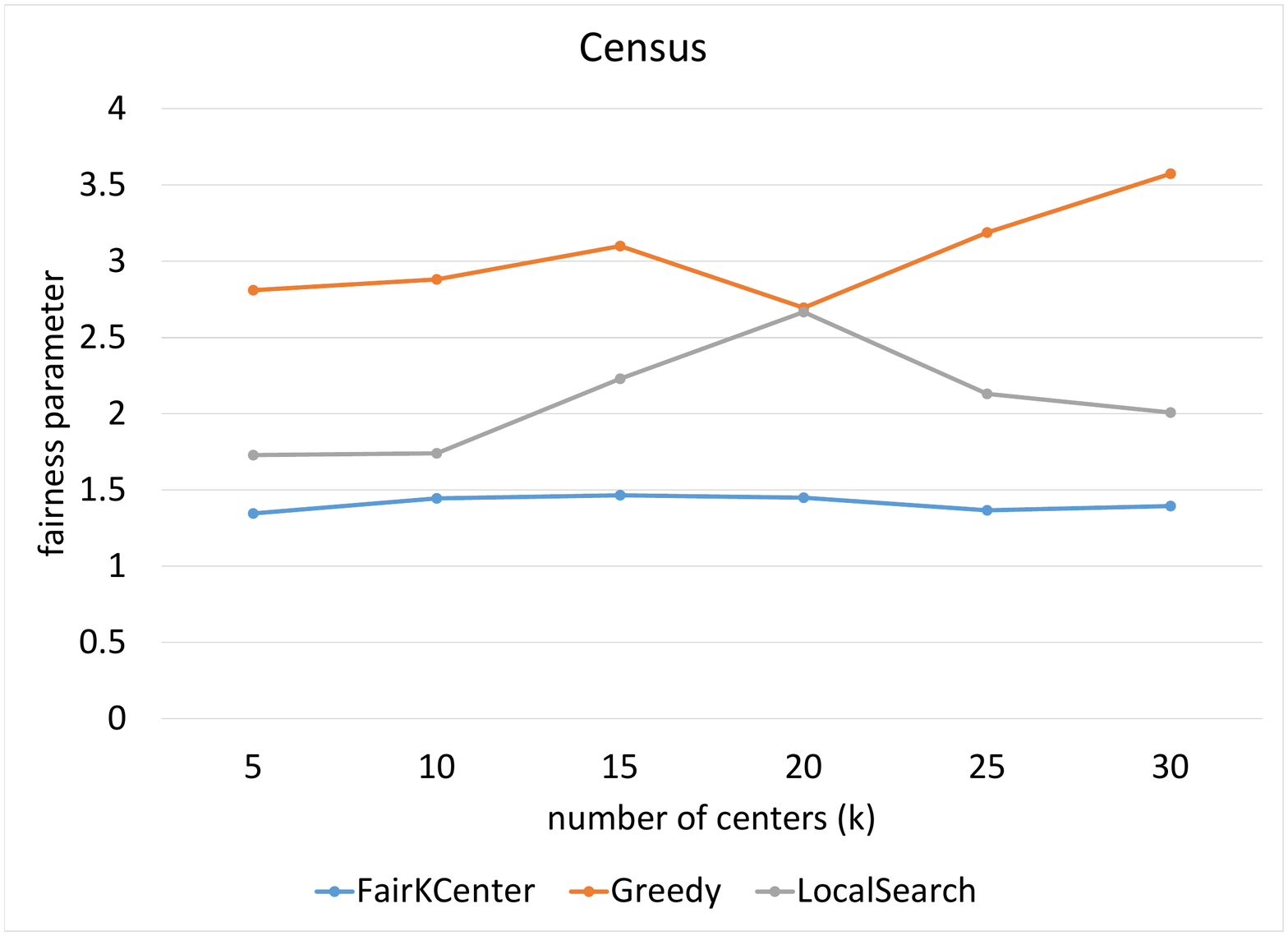}
\endminipage
\caption{Comparison of the fairness guarantees of the described algorithms for fair $k$-means on data sets Diabetes, Bank and Census.}
\label{fig:fair-means}
\end{figure*}

\begin{figure*}[!h]
\minipage{0.33\textwidth}
		\includegraphics[width=\textwidth]{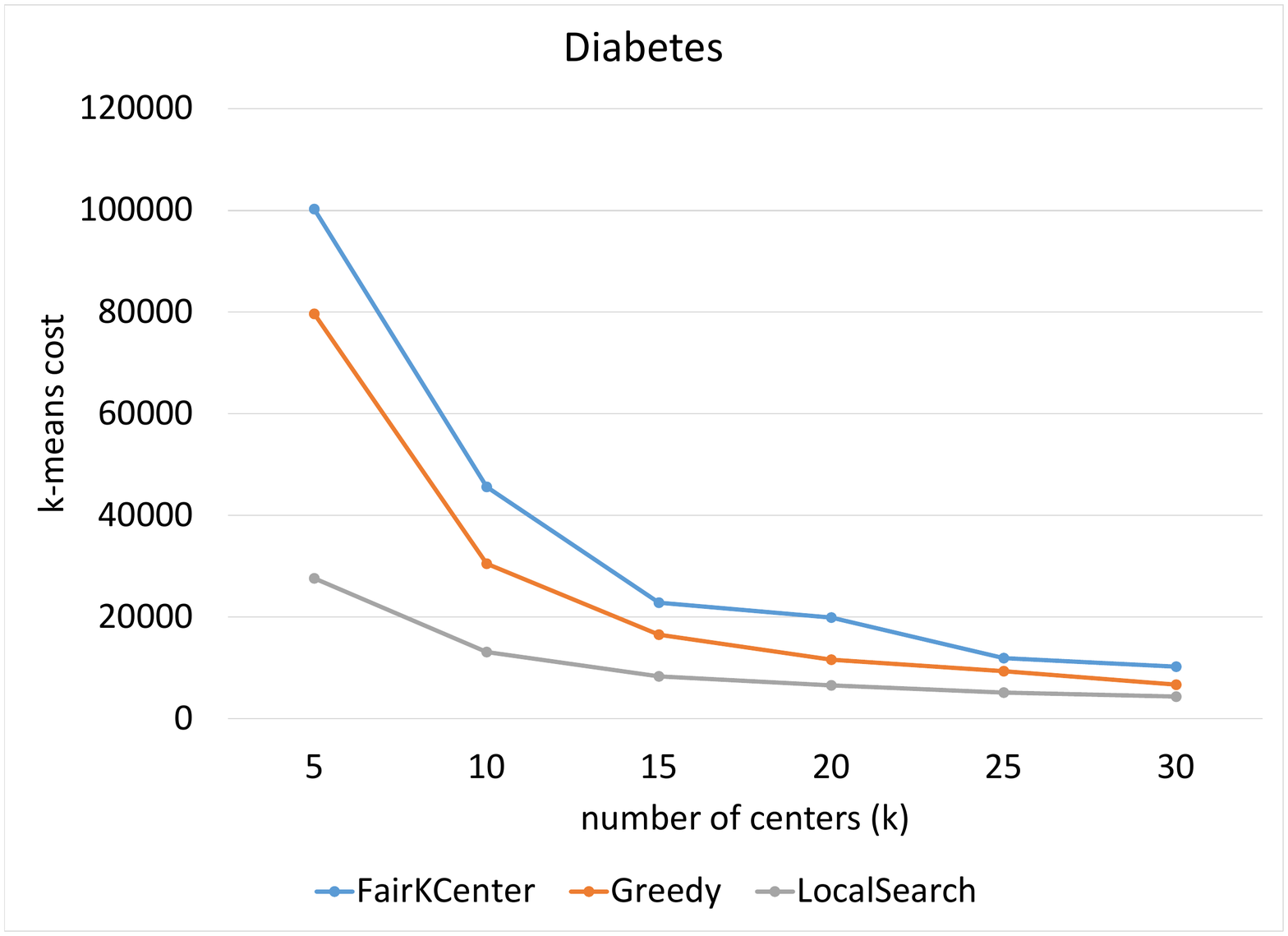}
\endminipage\hfill
\minipage{0.33\textwidth}
		\includegraphics[width=\textwidth]{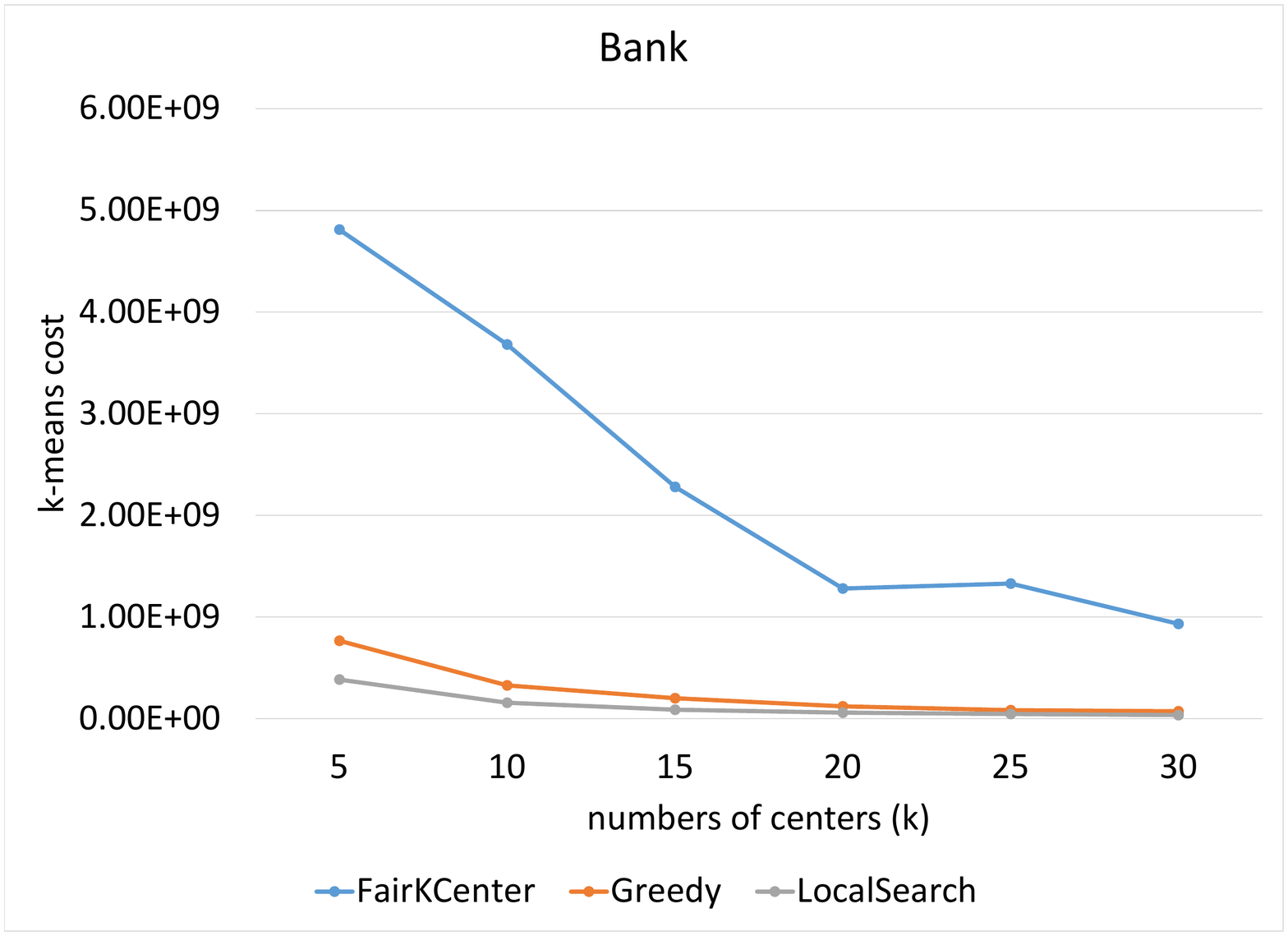}
\endminipage\hfill
\minipage{0.33\textwidth}%
		\includegraphics[width=\textwidth]{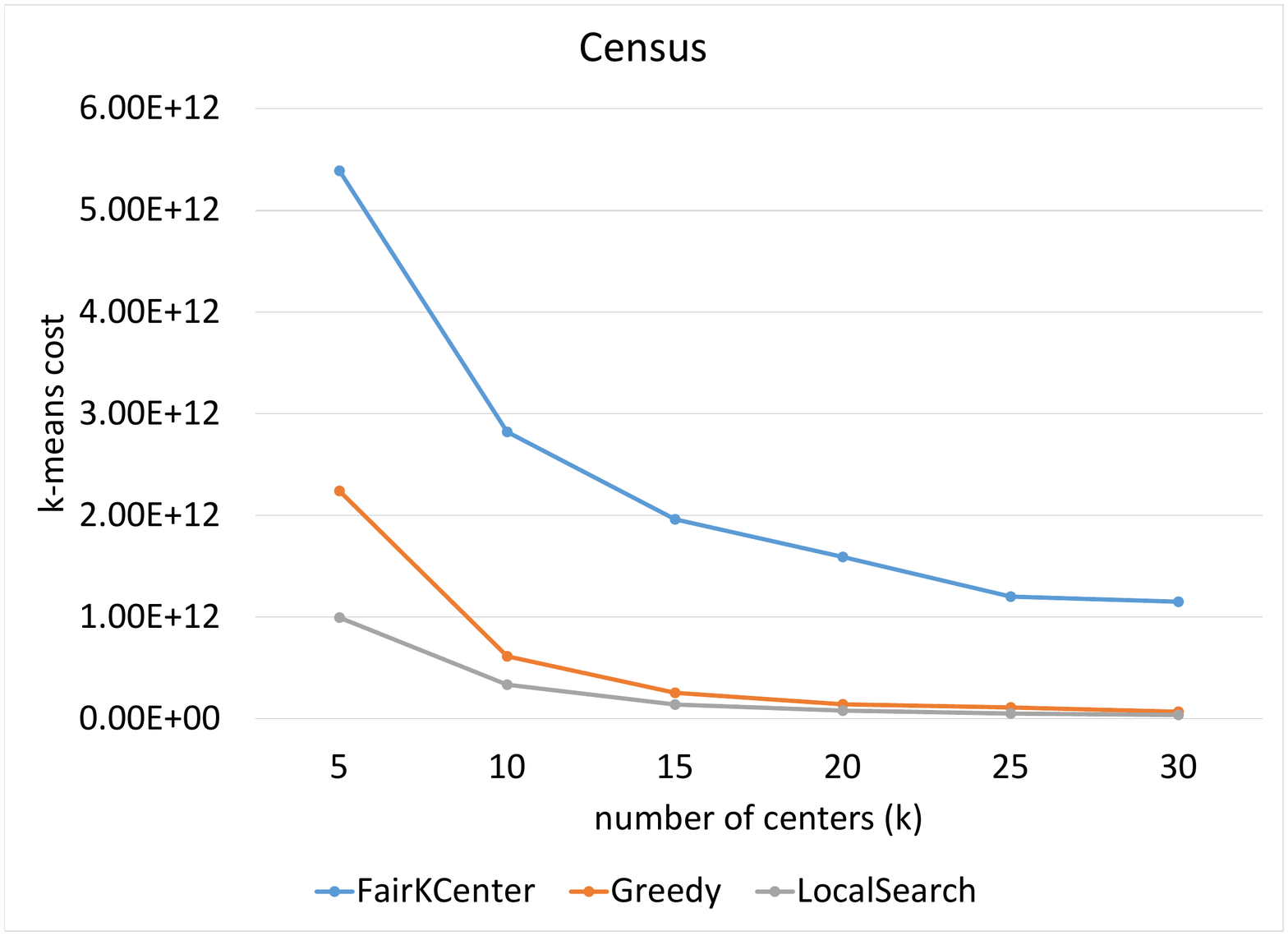}
\endminipage
\caption{Comparison of the $k$-means cost of the described algorithms for fair $k$-means on data sets Diabetes, Bank and Census.}
\label{fig:cost-means}
\end{figure*}

\paragraph{Results.} 
Figures \ref{fig:fair-median} and \ref{fig:cost-median} show empirical comparisons of the aforementioned algorithms, both in terms of fairness and the $k$-median cost of the solution.
Our plots imply that local-search based algorithms perform reasonably well with respect to the notion of $\alpha$-fairness: While its fairness guarantee is close to the fairness of \textsc{FairKCenter}, it always exhibits a better performance in terms of $k$-median cost. 
More precisely, on average the local search algorithm reports a solution whose cost is better than the algorithm of \cite{jung2019center} by a factor of 1.4, 2.25, and 1.93  
while the reported solution has a worse fairness by a factor of 1.13, 1.5, and 1.16 for Diabetes, Bank and Census respectively.

We also observe a similar behavior of local search based algorithm for fair $k$-means. 
Figures \ref{fig:fair-means} and \ref{fig:cost-means} show empirical comparisons of the aforementioned algorithms both in terms of fairness and the $k$-means cost of their solutions.
The plots show that local-search based algorithms perform reasonably well with respect to the notion of $\alpha$-fairness: While its fairness guarantee is close to the fairness of \textsc{FairKCenter}, it always exhibits a better performance in terms of $k$-means cost. 
More precisely, on average the local search algorithm reports a solution whose cost is better than the algorithm of \cite{jung2019center} by a factor of 2.93, 2.32, and 1.73  
while the reported solution has a worse fairness by a factor of 1.14, 1.85, and 1.48 for Diabetes, Bank and Census respectively.

\bibliographystyle{alpha}%
\bibliography{ref-fair-clustering}
\newpage
\appendix
\section{Missing Proofs}\label{sec:missing-proofs}
\subsection{Missing Proofs of Section~\ref{sec:intro}}\label{sec:miss-intro}
\begin{proofof}{\bf Observation~\ref{lem:optimal-is-unfair}.}
Consider the instance $\sI$ as shown in Figure~\ref{fig:costly-fair}. For simplicity, we assume that the input points $P$ live in a Euclidean $(k-1)$-dimensional space.\footnote{In fact, it is possible to modify our instance so that it lives in two dimensions. Even by applying standard dimensionality reduction techniques like Johnson-Lindenstrauss~\cite{johnson1984extensions}, we can preserve the distances in our construction approximately (up to a factor of $1\pm\eps)$) and reduce the number of dimensions to $O(\log k)$~\cite{cohen2015dimensionality, makarychev2019performance}.}
Suppose that there are $k-1$ points from $P$ on the left side where the distance of any pair of these $k-1$ points is exactly $M$. 
This can be achieved by picking the vertices of a standard simplex. Similarly, we pick $k-1$ ``nodes''\footnote{Note that \emph{nodes} are different from the actual points in $P$.} $\sV = \{v_1, \cdots, v_{k-1}\}$ on the right side such that the distance of any pair of nodes in $\sV$ is exactly $2R$. 
Then, we put ${n - k-1 \over k-1}-1 \geq {n\over k}-1$ (assuming $k^2\leq n$) points from $P$ on the ball of radius $r$ around each vertex in $\sV$ and exactly one point on each vertex in $\sV$. Further, we set the minimum distance of the points on right side and the points on the left side to $D$. 

Then, we choose the value of parameters so that $R>>r$ and $M=D >> 2(R+r)n$. Let $O$ be an optimal set of centers for $k$-median\footnote{A similar argument holds for $k$-means and the general $\ell_p$ norm cost function as well.} clustering of $P$. 
First we show that $O$ must include all points on the left side and exactly one point from the right side. If $O$ does not contain all points on the left side, then the total clustering cost is at least $M$, while in the described solution (i.e., picking all points on the left side and an arbitrary point from the right side) the total cost is at most $2(R+r)n << M$; since the distance of each point on the right side to its closest node in $\sV$ is $r$ and the distance of any pair of nodes in $\sV$ is exactly $2R$. 

Next, the fair radius of all points lied on $\sV$ is exactly $r$ and the fair radius of all points that live on balls of radius $r$ around the nodes in $\sV$ is at least $r$ and at most $2r$. 
Since the distance of at least one point on the right side to $O$ is at least $2R-2r$, the fairness approximation of $O$ is greater than ${2(R - r) \over 2r} = {R\over r} - 1$.

On the other hand, consider a set of centers $S$ that consists of the $k-1$ points lied on $\sV$ and one arbitrary point from the left side. Since $R>> r$, it is straightforward to check that all points both on the right side and on the left side have a center of $S$ in their fair radii. In particular, $S$ provides a $1$-fair $k$-clustering of $P$.

Thus, by setting $R\over r$ large enough, an optimal $k$-median clustering of $P$ can be arbitrarily unfair.    
\end{proofof}

\section{Fair Algorithms for the General $\ell_p$ Norm Cost Function}\label{sec:general-cost}
Our local search algorithm for $\alpha$-fair $k$-clustering with respect to the general $\ell_p$ norm cost function is similar to the one for $k$-median. For cost functions other than $k$-median, the local search algorithm was first analyzed with respect to the $k$-means cost function by~\cite{kanungo2004local} in Euclidean space using a so-called ``centroidal'' property of optimal solutions in $k$-means. Later, the analysis was both simplified and generalized by~\cite{gupta2008simpler}; their analysis showed that local search algorithm works for $k$-means and the more {\em general $\ell_p$ norm cost} function, $(\sum_{x\in P} d(x, S)^p)^{1/ p}$, in {\em general metric} spaces.

In this section, following the analysis of~\cite{gupta2008simpler}, we analyze our local search algorithm with respect to the $\ell_p$ norm cost function where $p\geq 1$. 
Note that this cost function has $k$-median (with $p=1$) and $k$-means (with $p=2$) as its special cases. Moreover, by setting $p=\log n$, it approximates the {\em $k$-center} cost function within a constant factor.  

We also note that, as in the analysis of fair $k$-median, we assume the existence of a $\Delta$-bounded mapping $\pi$ and $(t,\gamma)$-bounded covering of the edges of $\pi$, $\sQ$ (see Section~\ref{sec:mapping-covering} for more details on the bounded mapping and covering). With this assumption, we can show that if the local search algorithm stops at time $T$, the cost of the $k$-clustering of the point set $P$ using the centers at time $T$, $S$ is within a constant factor of the cost of an optimal $\alpha$-fair $k$-clustering of $P$ with respect to $\cost_p$. Note that by the termination condition of the local search algorithm, $S$ is a $(t,\eps)$-stable set of centers.
 
In the rest of this section, we use $\d(x,y)$ to denote $\d(x,y) := d(x,y)^p$.  
\begin{lemma}\label{lem:general-constant}
Consider a set of $n$ points $P$ in a metric space $(X, d)$. Let $O$ be a set of centers for an optimal $\alpha$-fair $k$-clustering of $P$ with respect to $\cost_p$ and let $S$ be a set of $(t, \eps)$-stable $k$ centers for which there exists a pair of $\Delta$-bounded mapping $\pi: O\rightarrow S$ and $(t, \gamma)$-bounded covering $\sQ$. Then, $\cost_p(S) \leq  16\gamma\cdot \Delta \cdot p \cdot \opt$ where $\opt$ denotes the cost of an optimal $\alpha$-fair $k$-clustering of $P$ with respect to $\cost_p$ (i.e., $\cost_p(O)$).  
\end{lemma}
\begin{proof}
Consider an arbitrary partition $Q \in \sQ$. 
Here, we bound the difference between the clustering cost of $P$ with the set $S$ as centers and the set $S_Q = (S\cup O(Q))\setminus{S(Q)}$ as centers. 

Let $S_j\subseteq P$ denote the cluster of points that are mapped to $s_j$ in the clustering with the set $S$ as centers and let $O_i\subseteq P$ denote the cluster of points that are mapped to $o_i$ in the clustering with the set $O$ as centers. Moreover, we use $\sO_Q$ and $\sS_Q$ respectively to denote $\bigcup_{o_i\in O(Q)} O_i$ and $\bigcup_{s_j\in S(Q)} S_j$.
\begin{align}
\sum_{x\in P} \d(x, S_Q) - \d(x, S) 
	&\leq \sum_{x\in \sO_Q} \d(x, S_Q) -  \d(x, S) + \sum_{x\in \sS_Q \setminus \sO_Q} \d(x, S_Q) - \d(x, S)  \nonumber \\
	&\leq \sum_{x\in \sO_Q} \d(x, o_x) - \d(x, s_x) \quad\rhd o_x=\nn_O(x), s_x = \nn_S(x) \nonumber \\ 
	&+ \sum_{x\in \sS_Q \setminus \sO_Q} \d(x, s_{o_x}) - \d(x, s_x) \quad\rhd s_{o_x}= \nn_S(o_x) \label{eq:general_cost_of_swap}
\end{align}
By summing over all partitions $Q\in \sQ$, 
\begin{align}
\sum_{Q\in \sQ} \sum_{x\in P} \d(x, S_Q) - \d(x, S) 
	&\leq \sum_{Q\in \sQ} \Big( \sum_{x\in \sO_Q} \big( \d(x, o_x) - \d(x, s_x)\big) + \sum_{x\in \sS_Q \setminus \sO_Q} \big( \d(x, s_{o_x}) - \d(x, s_x) \big) \Big) \nonumber\\ 
	&\leq \gamma\cdot \opt^p - \cost_p(S)^p \nonumber + \gamma\cdot \Delta \cdot (\sum_{x\in P} \d(x, s_{o_x}) - \d(x, s_x)) \nonumber \\
	&\leq \gamma\cdot \opt^p - (\gamma\cdot\Delta + 1) \cdot \cost_p(S)^p + \gamma\cdot \Delta \cdot \sum_{x\in P} \d(x, s_{o_x}) \label{eq:general-final-bound} 
\end{align}
where the second inequality follows from Property~D-\ref{enum:mapping-2}, Property~E-\ref{enum:partition-0} of $(\pi, \sQ)$ and the fact that given $S_j$, for all $x\in S_j$, $\d(x, s_{o_x}) - \d(x, s_j) \geq 0$. 
Next, we will bound the term $\sum_{x\in P} \d(x, s_{o_x})$ in Eq.~\eqref{eq:general-final-bound}.
\begin{claim}\label{clm:general-bound-R-term}
Let $s_{o_x}$ be the nearest center in $S$ to $o_x$ where $o_x$ is the nearest center in $O$ to $x$. Assuming $\beta = {\cost_p(S) / \opt}$, then $\sum_{x\in P} \d(x, s_{o_x}) \leq (2+\beta)^p\cdot \opt^p$.
\end{claim} 
\begin{proof}
Let $p_1, \cdots, p_n$ denote the points in $P$. We define vector $\X$ so that the $i$-th coordinate in $\X_i$ is equal to $d(p_i, o_{p_i})$. Similarly, we define $\Y$ so that $\Y_i = d(p_i, s_{p_i})$. Note that in particular, $\left\| 2\X \right\|_p = 2\cdot \opt$, $\left\| \Y \right\|_p = \cost_p(S)$ and all coordinates of $\X$ and $\Y$ are positive.
\begin{align*}
\sum_{x\in P} \d(x, s_{o_x}) 
&\leq \sum_{x\in P} (2d(x, o_x) + d(x, s_x))^p \;\rhd\text{Eq.~\eqref{eq:closest-stable-center}}\\
&= \left\|2\X + \Y\right\|_p^p\\
&\leq \big(\left\|2\X\right\|_p + \left\|\Y\right\|_p\big)^p \;\rhd\text{triangle inequality (Minkowski inequality on $L^p$ with $p\geq 1$)}\\
&\leq (2\cdot {\opt} + {\cost_p(S)})^p \\
&\leq (2+\beta)^p\cdot \opt^p \;\rhd\text{by $\beta = {\cost_p(S) / \opt}$}
\end{align*}
\end{proof}
Since $S$ is a $(t,\eps)$-stable set of centers and all partitions in $\sQ$ are of size at most $t$ (i.e., $|S(Q)|\leq t$), 
for each $Q\in \sQ$, $(\sum_{x\in P} \d(x, S_Q))^{1/p} \geq (1-\eps) \cdot \cost_p(S)$. Since for $\eps\in[0,1]$, $(1-\eps)^p \geq 1-p\eps$, it implies that,
\begin{align*}
\sum_{x\in P} \d(x, S_Q) - \d(x, S) = \sum_{x\in P} \d(x, S_Q) - \cost_p(S)^p \geq -p\cdot \eps \cdot \cost_p(S)^p.
\end{align*} 
Hence, together with Claim~\ref{clm:general-bound-R-term} and Eq.~\eqref{eq:general-final-bound}, 
\begin{align}
-p\cdot\eps |\sQ| \cost_p(S)^p  
&\leq \sum_{Q\in \sQ} \sum_{x\in P} \d(x, S_Q) - \d(x, S) \nonumber\\
&\leq \gamma\cdot \opt^p - (\gamma\cdot\Delta + 1) \cdot \cost_p(S)^p + \gamma\cdot \Delta \cdot \sum_{x\in P} \d(x, s_{o_x}) &&\rhd\text{by Eq.~\eqref{eq:general-final-bound}} \nonumber \\
&\leq \gamma\cdot \opt^p - (\gamma\cdot\Delta + 1) \cdot \cost_p(S)^p + \gamma\cdot \Delta \cdot (2+\beta)^p\cdot \opt^p\label{eq:general-bound}
\end{align}
Since $|\sQ| \leq k\cdot \gamma$, by setting $\eps =1/(2k\cdot \gamma\cdot p)$, $-p\cdot\eps |\sQ| \cost_p(S)^p \geq -\cost_p(S)^p/2$.
Hence, by rearranging Eq.~\eqref{eq:general-bound} and together with the assumption $\beta = {\cost_p(S) / \opt}$,
\begin{align}
\beta^p  = {\cost_p(S)^p \over \opt^p} \leq {\gamma (1 + \Delta\cdot (2+\beta)^p) \over {1\over 2} + {\gamma\cdot \Delta}}
\end{align}
which implies that 
\begin{align}
&\beta^p\cdot ({1\over 2} + \gamma\cdot \Delta) \leq \gamma + \gamma\cdot \Delta \cdot (2+\beta)^p \nonumber \\
\Rightarrow\; &{\beta^p}\cdot ({1\over 2} + \gamma\cdot \Delta\cdot (1- (1+{2\over \beta})^p) \leq \gamma \label{eq:general-formula}
\end{align}
\begin{claim}\label{clm:bound}
$1 - e^{1\over 8\gamma\cdot \Delta} \geq -{1\over 4\gamma\cdot \Delta}$.
\end{claim}
\begin{proof}
Since for $x\in(0,1)$, $\ln(1+x) \geq {x\over 2}$, $\ln(1+ {1\over 4\gamma\Delta}) \geq {1\over 8\gamma\Delta}$. Hence, since $\mathrm{exp()}$ is monotone, $1- e^{1\over 8\gamma\Delta} \geq -{1\over 4\gamma\Delta}$.
\end{proof}
Next, we show that $\beta \leq 2\cdot(8\gamma\cdot\Delta)\cdot p$. Suppose for contradiction that it is not the case. Then,
\begin{align*}
\beta^p \cdot ({1\over 2} + \gamma\cdot \Delta\cdot (1- (1+{2\over \beta})^p) 
&\geq \beta^p \cdot ({1\over 2} + \gamma\cdot \Delta\cdot(1 - (1+{1\over 8\gamma\cdot \Delta\cdot p})^p)) \\
&\geq \beta^p \cdot ({1\over 2} + \gamma\cdot \Delta\cdot(1 - e^{1\over 8\gamma\cdot \Delta})) &&\rhd (1+{1\over x})^x < e \\
&\geq \beta^p \cdot ({1\over 2} - {1\over 4}) &&\rhd\text{by Claim~\ref{clm:bound}}\\
&\geq \gamma &&\rhd\text{$p\geq 1$}
\end{align*}
which is a contradiction. Hence, $\beta = (\cost_p(S) / \opt) \leq 16\gamma\cdot\Delta\cdot p$ and the solution returned by the local search algorithm is an $O(p)$-approximation.
\end{proof}

\begin{theorem}\label{thm:general}
The local search algorithm with swaps of size at most $4$ returns a $(O(p),7)$-bicriteria approximate solution of $\alpha$-fair $k$-clustering of a point set $P$ of size $n$ with respect to the cost function $\cost_p$ in time $\tldO(pk^{5} n^4)$.
\end{theorem}
\begin{proof}
By Lemma~\ref{lem:clustering-with-partition-fair}, the result of our local search algorithm returns a $(7\alpha)$-fair $k$-clustering. By Lemma~\ref{lem:general-constant} and the existence of a pair of $3$-bounded mapping and $(4,6)$-bounded covering (see Lemma~\ref{lem:stable-mapping} and~\ref{lem:stable-partitioning}), the $\cost_p$ of the returned solution is $O(p\cdot \opt)$ where $\opt$ is the cost of an optimal $\alpha$-fair $k$-clustering of $P$ with respect to $\cost_p$. 

Finally, as we set $\eps = O({1\over p\cdot k})$ and by Theorem~\ref{thm:iterations} the runtime of the algorithm is $\tldO(pk^5 n^4)$.
\end{proof}

\begin{corollary}[restatement of Theorem \ref{thm:main} for $k$-means]\label{cor:main-means}
The local search algorithm with swaps of size at most $4$ returns a $(O(1),7)$-bicriteria approximate solution of $\alpha$-fair $k$-means of a point set $P$ of size $n$ in time $\tldO(k^{5} n^4)$.
\end{corollary}
\begin{proof}
It follows from Theorem~\ref{thm:general} by setting $p = 2$.
\end{proof}

\begin{corollary}[fair $k$-center]\label{cor:main-center}
The local search algorithm with swaps of size at most $4$ returns a $(O(\log n),7)$-bicriteria approximate solution of $\alpha$-fair $k$-center of a point set $P$ of size $n$ in time $\tldO(k^{5} n^4)$.
\end{corollary}
\begin{proof}
It follows from Theorem~\ref{thm:general} by setting $p = \log n$. Note that for any vector $\X\in \mathbb{R}^n$, $\left\|\X \right\|_{\infty} \leq \left\|\X \right\|_{\log n} \leq 2\cdot \left\|\X \right\|_{\infty}$.
\end{proof}

\end{document}